\documentclass[prx,twocolumn,floatfix,superscriptaddress,longbibliography]{revtex4-1}

\usepackage{amssymb,amsmath,amstext,amsthm,mathtools}
\usepackage{graphicx}
\usepackage{epstopdf}
\usepackage{color}
\usepackage{bm}
\usepackage{appendix}
\usepackage[T1]{fontenc}
\usepackage{bbold}
\usepackage{bbm}
\usepackage{latexsym}
\usepackage[colorlinks=true,citecolor=blue,linkcolor=blue]{hyperref}
\usepackage{float}
\usepackage{verbatim}
\usepackage{lipsum}
\usepackage{braket}
\usepackage[caption = false]{subfig}

\newcommand{\ketbra}[2]{| \hspace{1pt} #1 \rangle \langle #2 \hspace{1pt} |}

\newcommand{\dya}[1]{\ket{#1}\!\bra{#1}}
\newcommand{\ip}[2]{\langle #1|#2\rangle}      

\newcommand{\EC}{\mathcal{E}}

\newcommand{\OC}{\mathcal{O}}

\newcommand{\Tr}{{\rm Tr}}

\renewcommand{\geq}{\geqslant}
\renewcommand{\leq}{\leqslant}
\newcommand{\mte}[2]{\langle#1|#2|#1\rangle }

\renewcommand{\Re}{\text{Re}}

\renewcommand{\vec}[1]{\boldsymbol{#1}}  

\newcommand{\ad}{^\dagger}

\newcommand*{\id}{\openone}

\newcommand{\thv}{\vec{\theta}}
\newcommand{\gav}{\vec{\gamma}}
\newcommand{\alv}{\vec{\alpha}}

\newtheoremstyle{example}{\topsep}{\topsep}%
{}
{}
{\bfseries}
{:}
{   }
{\thmname{#1}\thmnumber{ #2}}
\theoremstyle{example}

\newtheorem{lemma}{Lemma}

\theoremstyle{definition}

\begin{document}

\title{Variational Fast Forwarding for Quantum Simulation Beyond the Coherence Time}

\author{Cristina C\^{i}rstoiu}
\thanks{The first two authors contributed equally to this work.}
\affiliation{Theoretical Division, Los Alamos National Laboratory, Los Alamos, NM, USA.}
\affiliation{Department of Computer Science, Oxford University, Oxford, UK.}

\author{Zo\"e Holmes}
\thanks{The first two authors contributed equally to this work.}
\affiliation{Theoretical Division, Los Alamos National Laboratory, Los Alamos, NM, USA.}
\affiliation{Information Sciences, Los Alamos National Laboratory, Los Alamos, NM, USA.}
\affiliation{Department of Physics, Imperial College, London, UK.}

\author{Joseph Iosue}
\affiliation{Theoretical Division, Los Alamos National Laboratory, Los Alamos, NM, USA.}
\affiliation{QC Ware Corporation, Palo Alto, CA, USA.}

\author{Lukasz Cincio}
\affiliation{Theoretical Division, Los Alamos National Laboratory, Los Alamos, NM, USA.}

\author{Patrick J. Coles}
\thanks{pcoles@lanl.gov}
\affiliation{Theoretical Division, Los Alamos National Laboratory, Los Alamos, NM, USA.}

\author{Andrew Sornborger} 
\thanks{sornborg@lanl.gov}
\affiliation{Information Sciences, Los Alamos National Laboratory, Los Alamos, NM, USA.}

\begin{abstract}
Trotterization-based, iterative approaches to quantum simulation are restricted to simulation times less than the coherence time of the quantum computer, which limits their utility in the near term. Here, we present a hybrid quantum-classical algorithm, called Variational Fast Forwarding (VFF), for decreasing the quantum circuit depth of quantum simulations. VFF seeks an approximate diagonalization of a short-time simulation to enable longer-time simulations using a constant number of gates. Our error analysis provides two results: (1) the simulation error of VFF scales at worst linearly in the fast-forwarded simulation time, and (2) our cost function's operational meaning as an upper bound on average-case simulation error provides a natural termination condition for VFF. We implement VFF for the Hubbard, Ising, and Heisenberg models on a simulator.  Additionally, we implement VFF on Rigetti's quantum computer to demonstrate simulation beyond the coherence time. Finally, we show how to estimate energy eigenvalues using VFF.
\end{abstract}
\maketitle

\section{Introduction}\label{sc:intro}

Quantum simulation (QS) was the earliest proposed example of a quantum algorithm that could outcompete the best classical algorithm \cite{feynman1982simulating}. Accelerated QS would impact fields including chemistry, materials science, and nuclear and high-energy physics. Current approaches include quantum emulation (or analogue QS)~\cite{JakschSim98, BlochSim12, MonroeSim2011, GormanSim18, kokail2019self}, Suzuki-Trotter-based methods~\cite{LloydQuantumSim, LloydFermiSim,SornborgerStewart99,campbell2019random}, and Taylor expansion-based QSs using linear combinations of unitaries~\cite{ChildsLCU,Somma2015, BabbushLCU}.  Quantum emulation and Suzuki-Trotter-based QSs have seen proof-of-principle demonstrations \cite{JakschSim98,BlochSim12,MonroeSim2011,GormanSim18, kokail2019self,FengEtAl2013}, while Taylor expansion-based QSs have the best asymptotic scaling and will likely have application for fault-tolerant Quantum Computers (QCs) of the future.

In the current noisy intermediate-scale quantum (NISQ) era, variational quantum simulation (VQS) methods are expected to be important. Variational algorithms have been introduced for finding ground and excited states~\cite{peruzzo2014VQE,McCleanPerturb, nakanishi2018subspace,Higgott2019} and for other applications~\cite{larose2018, arrasmith2019variational,Romero17, anschuetz2019variational}. In addition, some variational algorithms simulate system dynamics~\cite{li2017efficient, BenjaminVarSim2, BenjaminVarSim1, HeyaEtAl2019}. Of the variational dynamical simulation methods, some are based on knowledge of low-lying excited states~\cite{HeyaEtAl2019}, and some are iterative in time~\cite{BenjaminVarSim2, BenjaminVarSim1, li2017efficient}. Both approaches have the potential to outperform Suzuki-Trotter-based methods in the NISQ era.

Simulating the dynamics of a quantum system for time $T$ typically requires $\Omega(T)$ gates so that a generic Hamiltonian evolution cannot be achieved in sublinear time. This result is known as the `No Fast Forwarding Theorem' and holds both for a typical unknown Hamiltonian \cite{atia2017fast} and for the query model setting \cite{BerryEtAl2007}. However, there are particular Hamiltonians that can be fast forwarded, which means that the quantum circuit depth does not need to grow significantly with simulation time. Hamiltonians that allow fast-forwarding are precisely those that lead to violations of time-energy uncertainty relations and equivalently allow for precise energy measurements~\cite{atia2017fast}. For example, commuting local Hamiltonians \cite{atia2017fast}, quadratic fermionic Hamiltonians \cite{atia2017fast}, and continuous-time quantum walks on particular graphs \cite{loke2017efficient} can all be fast forwarded. In addition, Ref.~\cite{PhysRevA.79.032316} exploited the exact solvability of the transverse Ising model to formulate a quantum circuit for its exact diagonalization, allowing for fast forwarding. This circuit was used to simulate the Ising model on Cloud QCs~\cite{CerveraLierta2018}. A subspace-search variational eigensolver was employed in \cite{HeyaEtAl2019} to fast forward low-lying states in a quantum system. In \cite{novo2019quantum} a Hamiltonian whose diagonalization is constructed out of IQP circuits is shown to give a quantum advantage for the task of energy sampling. More generally, it remains an open problem to determine the precise form for Hamiltonians that do and do not allow fast forwarding. 

Previous results analyze fast forwarding of Hamiltonians mostly in a computational complexity setting \cite{atia2017fast,BerryEtAl2007,Childs:2010:LSN:2011373.2011380} in which the asymptotic scaling of the runtime of quantum circuits implementing a large scale simulation is important. However, near-term devices are constrained to simulating intermediate scale systems using finite depth circuits. 
The behavior of an algorithm to simulate large systems and long times may not be indicative of its behavior in smaller scale regimes. 
Therefore, as discussed further in \color{black}{\ref{ap:FastForwarding}}, whether or not asymptotic fast-forwarding is possible for a particular Hamiltonian has limited impact on the simulations that may be performed using near-term QCs. 

The advantage of fast forwarding, if possible, for near-term QCs is that the simulation time $T$ can be much longer than the coherence time $\tau$ of the QC performing the simulation. This is because $T$ is just a parameter that is set `by hand' in a fixed-depth quantum circuit \cite{PhysRevA.79.032316,HeyaEtAl2019}. 
Therefore we ask the following core question: \emph{Can we fast forward the evolution of a Hamiltonian beyond the coherence time of a near-term device using a variational algorithm?}

\begin{figure}[t]
\includegraphics[width=\columnwidth]{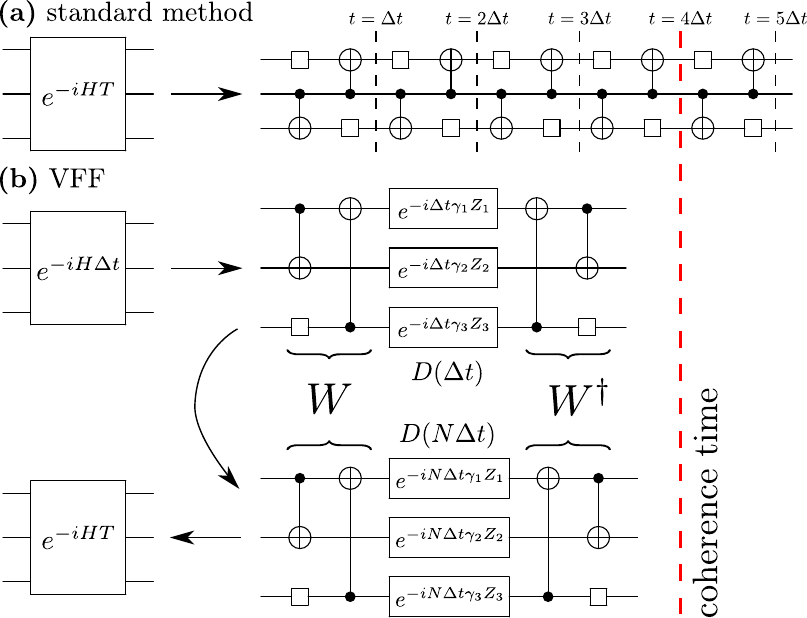}
\caption[]{The concept of Variational Fast Forwarding (VFF). (a) A Trotterization-based quantum simulation with $N=5$ timesteps. This simulation runs past the coherence limit of the quantum architecture. (b) A VFF-based quantum simulation. An approximate diagonalization of a short-time simulation is found variationally. Using the eigenvector $W$ and diagonal $D$ unitaries that were learned, an arbitrary length simulation is implemented by modifying the parameters in $D$. As long as VFF results in few enough gates that the circuit does not exceed the coherence time, longer simulations can be performed than the standard method in (a).}
\label{fig:Decomp}
\end{figure}
In this paper, we introduce a variational, hybrid quantum-classical algorithm that we call Variational Fast Forwarding (VFF). We envision it to be most useful for implementing quantum simulations on near-term, NISQ computers. However, it could also have uses in fault-tolerant QS. It is distinct from SVQS \cite{HeyaEtAl2019} in that our method searches for an approximate diagonalization of an entire QS unitary, rather than for a finite set of low-lying states. Most importantly, we analyze the simulation errors produced by VFF and guarantee a desired accuracy for the simulation once a termination condition is achieved. This is possible due to the operational meaning of our cost function. In contrast, low-energy subspace approaches as in SVQS may not be able to guarantee a desired simulation error, since the cost function (i.e., the energy) does not carry an obvious operational meaning.

The basic idea of VFF is depicted in Fig.~\ref{fig:Decomp}. Section~\ref{sec:Results} presents our main results including an overview of the algorithm, the cost function, error analysis, and implementations of VFF on a simulator and on Rigetti's QC. Section~\ref{sec:Discussion} discusses these results, and Section~\ref{sec:Methods} elaborates on our ansatz and training methods. 

\section{Results}\label{sec:Results}

\subsection{The VFF Algorithm}

\subsubsection{Overview}

\begin{figure*}
\begin{center}
\includegraphics[width=1.8\columnwidth]{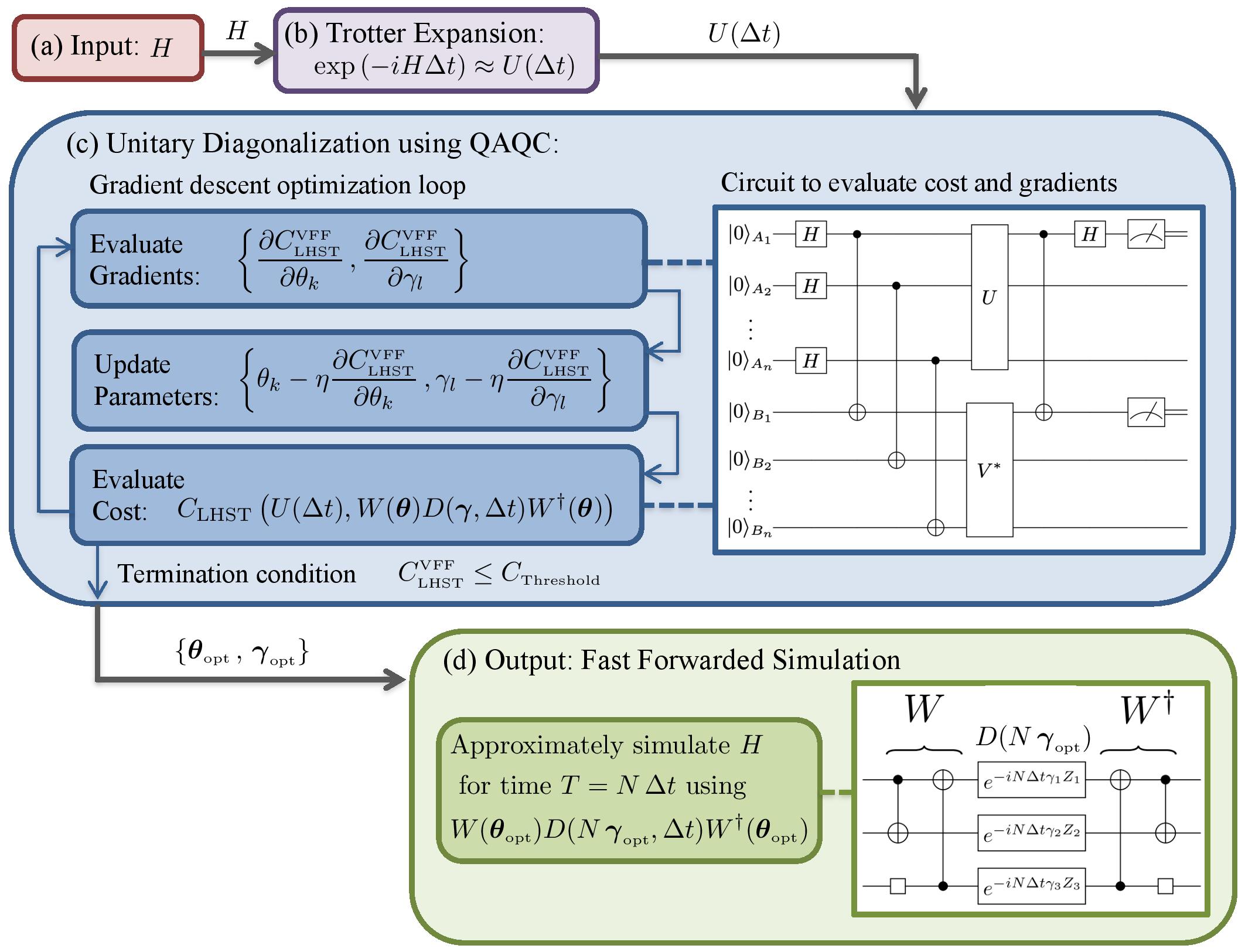}
\makeatletter
\renewcommand{\@makecaption}[2]{%
  \par\vskip\abovecaptionskip\begingroup\small\rmfamily
  \splittopskip=0pt
  \setbox\@tempboxa=\vbox{
    \@arrayparboxrestore \let \\\@normalcr
    \hsize=.5\hsize \advance\hsize-1em
    \let\\\heading@cr
    \@make@capt@title {#1}{#2}
  }%
  \vbadness=10000
  \setbox\z@=\vsplit\@tempboxa to .55\ht\@tempboxa
  \setbox\z@=\vtop{\hrule height 0pt \unvbox\z@}
  \setbox\tw@=\vtop{\hrule height 0pt \unvbox\@tempboxa}
  \noindent\box\z@\hfill\box\tw@\par
  \endgroup\vskip \belowcaptionskip
}
\makeatother
\end{center}
\caption{The VFF Algorithm. (a) An input Hamiltonian is transformed into (b) a gate sequence associated with a single-timestep Trotterized unitary, $U(\Delta t)$. (c) The unitary is then variationally diagonalized by fitting a parameterized factorization, $V(\alv, \Delta t) = W(\thv) D(\gav , \Delta t) W^\dagger(\thv)$. This variational subroutine employs gradient descent to minimize a cost function $C_{\mbox {\tiny LHST}}$, whose gradient is efficiently estimated with a short-depth quantum circuit called the Local Hilbert-Schmidt Test (LHST). The variational loop is exited when a termination condition given by \eqref{eqn:termination} is reached, which guarantees that a user-defined bound on the average fidelity $\overline{F}(T)$ is achieved. (d) After the termination condition is reached, the optimal parameters $(\thv_{\text{opt}}, \gav_{\text{opt}})$ are used to implement a fast-forwarded simulation, with the fast-forwarding error growing sub-linearly in the simulation time (see Eq.~\eqref{eqnTriangleInequality2}). The fast-forwarding is performed by modifying the parameters of the diagonal unitary, $D(\gav_{\text{opt}} , \Delta t) \rightarrow D(N \gav_{\text{opt}} , \Delta t)$, producing a quantum simulation unitary, $W(\thv_{\text{opt}}) D(N \gav_{\text{opt}} , \Delta t) W^\dagger(\thv_{\text{opt}})$. 
}
\label{fig:VFF-algorithm}
\end{figure*}

Given a Hamiltonian $H$ on a $d=2^n$ dimensional Hilbert space (i.e., on $n$ qubits) evolved for a short time $\Delta t$ with the simulation unitary $e^{-iH\Delta t}$, the goal is to find an approximation that allows the simulation at later times $T$ to be fast forwarded beyond the coherence time.
Figure~\ref{fig:VFF-algorithm} schematically shows the VFF algorithm, which consists of the following steps:
\begin{enumerate}
	\item Implement a unitary circuit $U(\Delta t)$ to approximate $e^{-iH\Delta t}$, the simulation at a small time step. 
	\item Compile $U(\Delta t)$ to a diagonal factorization $V = WDW^\dagger \approx e^{-iH\Delta t}$ with circuit depth $L$.
	\item Approximately fast forward the quantum simulation at large time $T = N \Delta t$ using the same circuit of depth $L$: $e^{-iH T} \approx W D^N W^\dagger$.
\end{enumerate} 

Typically $U(\Delta t)$ will be a single-timestep Trotterized unitary approximating $e^{-iH\Delta t}$. We variationally search for an approximate diagonalization of $U(\Delta t)$ by compiling it to a unitary with a structure of the form 
\begin{align}
\label{eqn:Vansatz}
    V(\alv , \Delta t) := W(\thv) D(\gav , \Delta t) W(\thv)^\dagger\,,
\end{align}
with $\alv = (\thv , \gav)$ being a vector of parameters. Here, $D(\gav , \Delta t)$ is a parameterized unitary composed of commuting unitaries that encode the eigenvalues of $U(\Delta t)$ while $W(\thv)$ is a parameterized unitary matrix consisting of corresponding eigenvectors. In Sec.~\ref{sec:Methods}, we describe layered structures that provide ans\"atze for the circuits $W(\thv)$ and $D(\gav , \Delta t)$, and we detail our gradient-descent optimization methods for training $\thv$ and $\gav$. 

To approximately diagonalize $U(\Delta t)$, the parameters $\alv = (\thv , \gav)$ are variationally optimized to minimize a cost function $C_{\mbox {\tiny LHST}}(U(\Delta t), V)$ that can be evaluated using a short-depth quantum circuit called the Local Hilbert-Schmidt Test (LHST)~\cite{khatrietal2019} shown in Fig.~\ref{fig:VFF-algorithm}(c). The compilation procedure we employ to approximate $U(\Delta t)$ by $V(\alv , \Delta t)$ makes use of the Quantum-Assisted Quantum Compiling (QAQC) algorithm \cite{khatrietal2019}, that was later shown to be robust to quantum hardware noise \cite{SharmaEtAl2019}. Section~\ref{sec:cost} below elaborates on our cost function.

If we can find such an approximate diagonalization for $U(\Delta t)$ then, for any total simulation time, $T = N\Delta t$, we have: 
\begin{align}
e^{-i H T} & =  (e^{-i H \Delta t})^N\\
\label{eqnApprox1}& \approx  (U(\Delta t))^N\\
\label{eqnApprox2}& \approx  W(\thv) D(\gav , \Delta t)^N W(\thv)^\dagger\\
& =   W(\thv) D(N \gav, \Delta t) W(\thv)^\dagger \,. \label{VFF}
\end{align}
Hence, a QS for any total time, $T$, may be performed with a fixed quantum circuit structure as depicted in Fig.~\ref{fig:VFF-algorithm}(d). In Sec.~\ref{sec:ErrorAnalysis}, we perform an error analysis to investigate how the approximate equalities in \eqref{eqnApprox1} and \eqref{eqnApprox2} affect the overall simulation error.

\subsubsection{Cost Function and Cost Evaluation}\label{sec:cost}

For the variational compiling step of VFF (shown in Fig.~\ref{fig:VFF-algorithm}(c)), we employ the cost function $C_{\mbox {\tiny LHST}}(U, V)$ introduced in Ref.~\cite{khatrietal2019}. This is defined as
\begin{equation}\label{eq-LHST_1}
    C_{\mbox {\tiny LHST}}(U, V)= 1- \frac{1}{n}\sum_{j=1}^n F_e^{(j)},
\end{equation}
where the $F_e^{(j)}$ are entanglement fidelities and hence satisfy $0\leq F_e^{(j)} \leq 1$. Specifically, $F_e^{(j)}$ is the entanglement fidelity for the quantum channel obtained from feeding into the unitary $U V\ad$ the maximally mixed state on $\overline{j}$ and then tracing over $\overline{j}$ at the output of $U V\ad$, where $\overline{j}$ contains all qubits except for the $j$-qubit. We elaborate on the form of $C_{\mbox {\tiny LHST}}(U, V)$ in \ref{ap:Cost}.

This function has several important properties. 
\begin{enumerate}
    \item It is faithful, vanishing if and only if $V = U$ (up to a global phase).
    \item Non-zero values are operationally meaningful. Namely, $C_{\mbox {\tiny LHST}}(U, V)$ upper bounds the average-case compilation error as follows:
\begin{equation}
    C_{\mbox {\tiny LHST}}(U, V) \geq  \frac{d+1}{nd}(1 - \overline{F}(U,V))\,,
\end{equation}
where $\overline{F}(U,V)$ is the average fidelity of states acted upon by $V$ versus those acted upon by $U$, with the average being over all Haar-measure pure states.
    \item The cost function appears to be trainable, in the sense that it does not have an obvious barren plateau issue (i.e., exponentially vanishing gradient, see Refs.~\cite{khatrietal2019, BarrenPlateausCerezo}).
    \item Estimating the cost function is DQC1-hard and hence it cannot be efficiently estimated with a classical algorithm~\cite{khatrietal2019}.
    \item There exists a short-depth quantum circuit for efficiently estimating the cost and its gradient.
\end{enumerate}
\noindent Regarding the last point, each $F_e^{(j)}$ term in \eqref{eq-LHST_1} is estimated with a different quantum circuit and then one classically sums them up to compute $C_{\mbox {\tiny LHST}}(U,V)$. An example of such a circuit is depicted in Fig.~\ref{fig:VFF-algorithm}(c). It involves $2n$ qubits, with the top (bottom) $n$ qubits denoted $A$ ($B$). The probability of the 00 measurement outcome on qubits $A_jB_j$ in this circuit is precisely the entanglement fidelity $F_e^{(j)}$. Therefore $2n$ single qubit measurements are required to compute $C_{\mbox {\tiny LHST}}(U,V)$, a favourable scaling compared to, for example, the $O(n^4)$ measurements that are naively required for VQE~\cite{peruzzo2014VQE, VQEMeasurementsVerteletskyi,VQEMeasurementsGokhale, VQEMeasurementsCrawford}. We also remark that $C_{\mbox {\tiny LHST}}(U,V)$ was recently shown to have noise resilience properties, in that noise acting during the circuit in Fig.~\ref{fig:VFF-algorithm}(c) tends not to affect the global optimum of this function~\cite{SharmaEtAl2019}.

For simplicity, we will often write our cost function as
\begin{equation}
\label{CLHSTVFFdef}
    C^{\mbox{\tiny VFF}}_{\mbox{\tiny LHST}}({T}) := C_{\mbox{\tiny LHST}}(U^{\frac{{T}}{ \Delta t}}, V^{\frac{{T}}{ \Delta t}})
\end{equation}
with $U = U(\Delta t)$ and $V = V(\alv, \Delta t)$, and note that $C^{\mbox{\tiny VFF}}_{\mbox{\tiny LHST}}(\Delta t)$ is the quantity that we directly minimize in the optimization loop of VFF.

\subsubsection{Simulation Error Analysis}\label{sec:ErrorAnalysis}

\textit{Linear scaling in $N$.---}In practice, each of the steps in the VFF algorithm above will generate errors. This includes the algorithmic error from the approximate implementation, $U(\Delta t)$, of the infinitesimal time evolution operator $e^{-iH\Delta t}$ and error from the approximate compilation and diagonalization of $U(\Delta t)$ into $V(\alv , \Delta t)$. These two error sources bound the overall error via the triangle inequality:
\begin{align}
\label{eqnTriangleInequality1}
    \epsilon^{\mbox{\tiny FF}}_p(\Delta t) \leq \epsilon^{\mbox{\tiny TS}}_p(\Delta t) + \epsilon^{\mbox{\tiny ML}}_p(\Delta t)\,.
\end{align}
Here, $\epsilon^{\mbox{\tiny FF}}_p(\Delta t)$ is the overall simulation error for time $\Delta t$, $\epsilon^{\mbox{\tiny TS}}_p(\Delta t)$ is the Trotterization error (note that this error may always be reduced using higher-order Trotterizations at the cost of more gates), and $\epsilon^{\mbox{\tiny ML}}_p(\Delta t)$ is the ``machine learning'' error associated with the variational compilation step. These quantities are defined as
\begin{align}
\epsilon^{\mbox{\tiny FF}}_p(\Delta t) & =  \| e^{-iH \Delta t} - V(\alv , \Delta t) \|_{p} \\
\epsilon^{\mbox{\tiny TS}}_p(\Delta t)  & =  \| e^{-iH \Delta t} - U(\Delta t) \|_p \\
\epsilon^{\mbox{\tiny MS}}_p(\Delta t)  & =  \| U(\Delta t) - V(\alv , \Delta t) \|_p \,,
\end{align}
where $\|M\|_p = (\sum_j m_j^p)^{1/p}$ is the Schatten $p$-norm, with $\{m_j\}$ the singular values of $M$. 

Ultimately we are interested in fast-forwarding and hence we want to bound $\epsilon^{\mbox{\tiny FF}}_p(T)$ with $T = N\Delta t$. For this purpose, we prove a lemma in \ref{sec:SimulationErrors} stating that
\begin{equation}
\label{eqnLemmaNsteps}
\| U_1^N - U_2^N \|_p \leq N \| U_1 - U_2 \|_p \,,
\end{equation}
for any two unitaries $U_1$ and $U_2$. Combining this lemma with the triangle inequality in \eqref{eqnTriangleInequality1} gives
\begin{align}
\label{eqnTriangleInequality2}
    \epsilon^{\mbox{\tiny FF}}_p(T) \leq N (\epsilon^{\mbox{\tiny TS}}_p(\Delta t) + \epsilon^{\mbox{\tiny ML}}_p(\Delta t))\,.
\end{align}
Equation~\eqref{eqnTriangleInequality2} implies that the overall simulation error scales at worst linearly with the number of time steps, $N$.

We remark that, for the special case of $p=2$, Eq.~\eqref{eqnLemmaNsteps} can be reformulated in terms of our cost function as:
\begin{equation}
\label{eqn:subquadratic}
    	C^{\mbox{\tiny VFF}}_{\mbox{\tiny LHST}}(T) \lessapprox  n \, N^2 \, C^{\mbox{\tiny VFF}}_{\mbox{\tiny LHST}}(\Delta t) \,, 
\end{equation}
with $C^{\mbox{\tiny VFF}}_{\mbox{\tiny LHST}}$ given by \eqref{CLHSTVFFdef}. The approximation in \eqref{eqn:subquadratic} holds when the cost function $C^{\mbox{\tiny VFF}}_{\mbox{\tiny LHST}}(\Delta t )$ is small, which is the case after a successful optimization procedure. See \ref{sec:SimulationErrors} for the non-approximate version of \eqref{eqn:subquadratic}. Thus we find that the VFF cost function scales at worst quadratically in $N$ under fast forwarding. 

\bigskip

\textit{Certifiable error and a termination condition.---}Equation~\eqref{eqnTriangleInequality2} holds for all Schatten norms, but of particular interest for our purposes is the Hilbert-Schmidt norm, $p =2$, from which we can derive \textit{certifiable} error bounds on the \textit{average-case} error. In addition, the operator norm, $p = \infty$, quantifies the \textit{worst-case} error and is often used in the quantum simulation literature~\cite{Wecker2014, Poulin2015}. For our numerical implementations (Section~\ref{sec:Implementations}), we will consider both worst-case and average-case error. On the other hand, for our analytical results presented here, we will focus on average-case error since it is naturally suited to providing a termination condition for the optimization in VFF.

As an operationally-meaningful measure of average-case error we consider the average gate fidelity between the target unitary $e^{-iHT}$ and the approximate simulation $V(\alv, T)$ arising from the VFF algorithm:
\begin{equation}
\label{eqn:AverageFidelity}
\overline{F}(T) =  \int_\psi | \mte{\psi}{V(\alv, T)\ad e^{-iHT}} |^2 d\psi, 
\end{equation}
where the integral is over all states $\ket{\psi}$ chosen according to the Haar measure.

In \ref{sec:SimulationErrors} we show that one can lower bound $\overline{F}(T)$ based on the value of the VFF cost function,
\begin{align}
\label{eqn:approxLowerBoundF}
\overline{F}(T) \gtrapprox 1 - \frac{d}{d+1} N^2 \left(\epsilon_{\infty}^{\mbox{\tiny TS}}(\Delta t)  + \sqrt{ n C_{\mbox{\tiny LHST}}^{\mbox{\tiny VFF}}(\Delta t)  } \right)^2   \,.
\end{align}
This inequality holds to a good approximation in the limit that $C_{\mbox{\tiny LHST}}^{\mbox{\tiny VFF}}(\Delta t)$ is small, as is the case after a successful optimization procedure. See \ref{sec:SimulationErrors} for the exact lower bound on $\overline{F}(T)$, from which \eqref{eqn:approxLowerBoundF} is derived.

\begin{figure}[!t]
\centering
\includegraphics[width=\columnwidth]{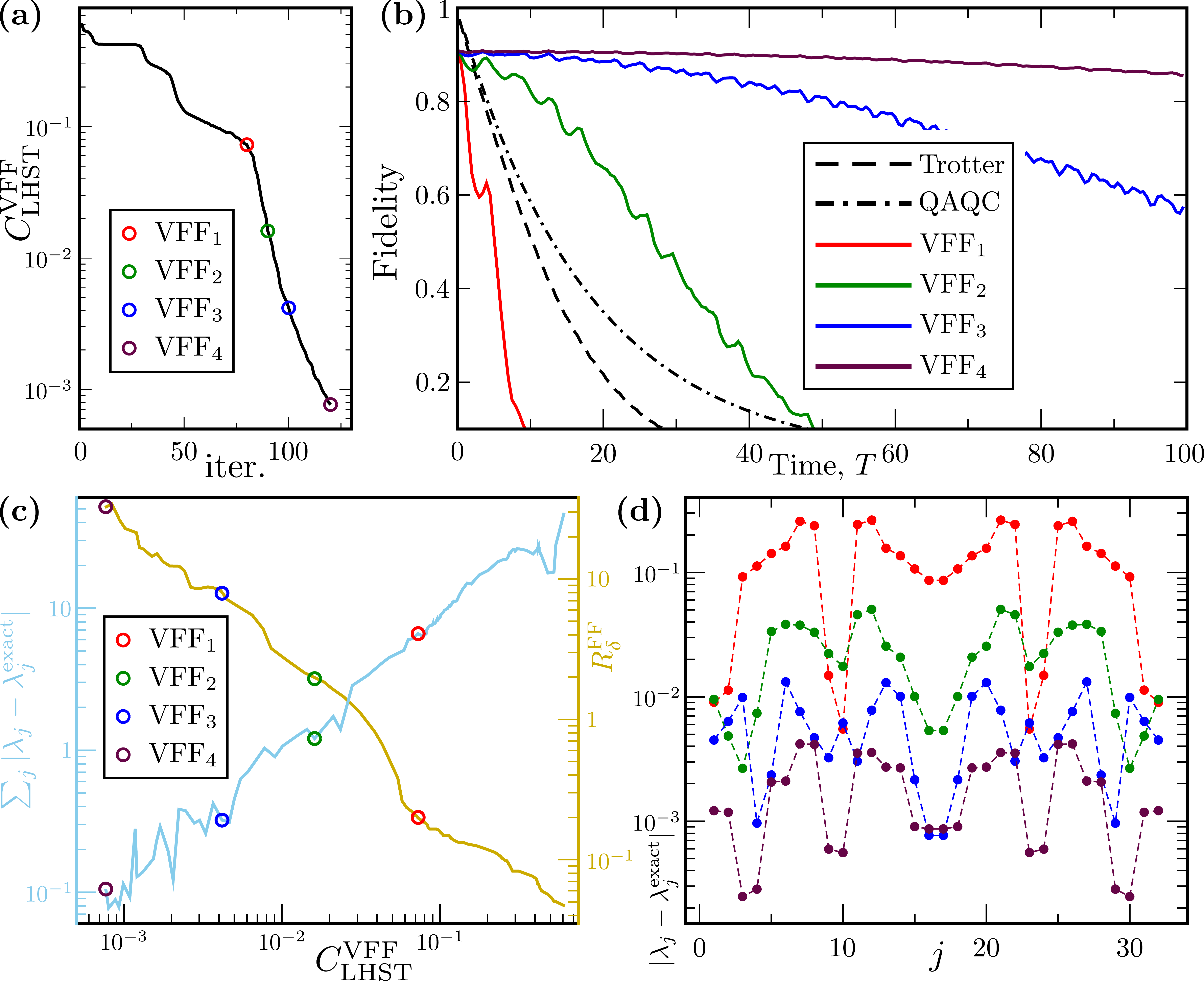}
\caption{\label{fig:Compare} Comparing different quantum simulation strategies. (a) The VFF cost function as it is iteratively minimized. Various quality diagonalizations are indicated in the colored circles. (b) Simulation fidelity as a function of time simulated. (c) Summed eigenenergy error  and fast forwarding as a function of the VFF cost. (d) Eigenenergy errors for the set of diagonalizations.}
\end{figure}

\begin{figure*}[t!]
    \centering
    \includegraphics[width=0.98\textwidth]{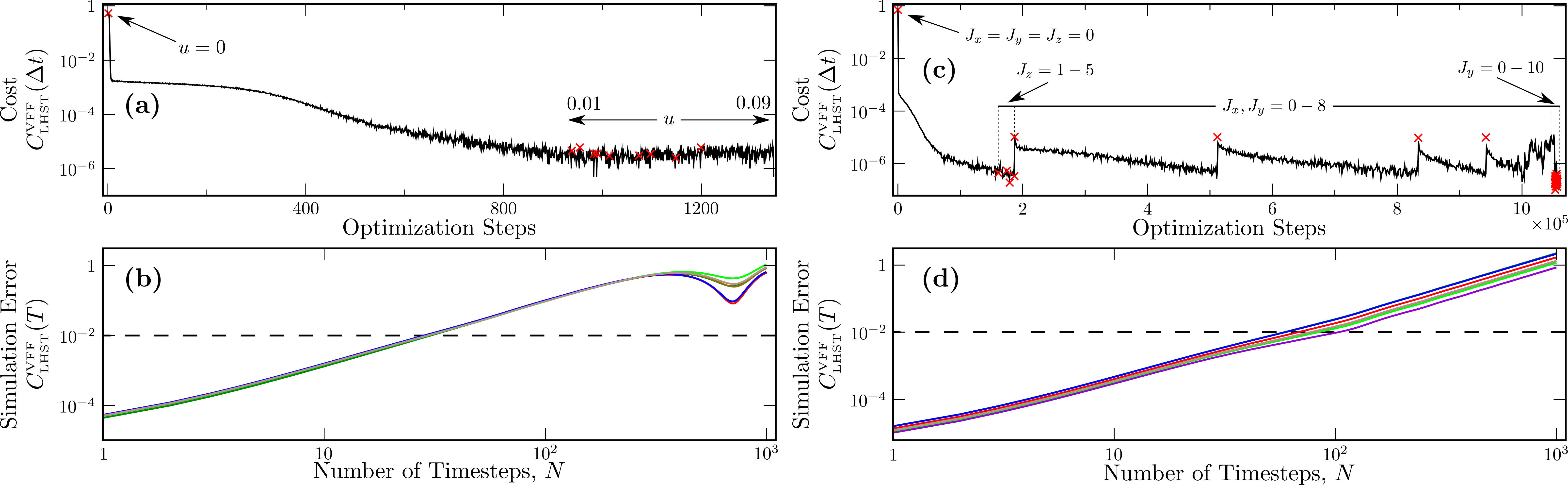}
    \caption{Finding successive diagonalizations across a range of parameters. (a,b) VFF of a two-site, two-qubit Hubbard quantum simulation unitary (4-qubit circuit). (a) Optimization error. Here, cost estimates were made with $n_\mathrm{samp} = 10^{6}$ and $\Delta t = 0.1$. We plot $C^{\mbox {\tiny VFF}}_{\mbox {\tiny LHST}}(\Delta t)$ versus optimization step for a sequence of parameters (see text). In this plot, red crosses depict the initial costs for each parameter before optimization. Each optimization was terminated after reaching $C_\mathrm{Threshold} = 10^{-6}$. After taking some time to diagonalize the initial unitary with $u = 0$, subsequent optimizations took just a few iterations. (b) Simulation error. Here, we plot $C^{\mbox {\tiny VFF}}_{\mbox {\tiny LHST}}(T)$ versus $N$ for all $u$. 
    For this level of optimization, and taking $T^\mathrm{Trot}_\delta$ to be one Trotter step, fast forwardings of approximately $30$ timesteps were achieved.
    (c,d) VFF of a three-qubit Heisenberg quantum simulation unitary (6-qubit circuit). (c) Optimization error. Estimates were made with $n_\mathrm{samp} = 10^{6}$ and $\Delta t = 0.1$. We plot $C^{\mbox {\tiny VFF}}_{\mbox {\tiny LHST}}(\Delta t)$ versus optimization step for a sequence of parameters (see text). In this plot, red x's depict the initial costs for each parameter before optimization. Each optimization was terminated after reaching $C_\mathrm{Threshold} = 10^{-6}$. (d) $C^{\mbox {\tiny VFF}}_{\mbox {\tiny LHST}}(T)$ versus $N$ plotted for all values of $J_z$, $J_x$, and $J_y$. Here, fast forwardings of approximately $70$ to $100$ timesteps were achieved.}
    \label{fig:Hubbard}
\end{figure*}

In addition, Eq.~\eqref{eqn:approxLowerBoundF} provides a termination condition for the variational portion of VFF. If one has a desired threshold for $\overline{F}(T)$, then this threshold can be guaranteed provided that $C_{\mbox{\tiny LHST}}^{\mbox{\tiny VFF}}(\Delta t)$ is below a certain value. Once $C_{\mbox{\tiny LHST}}^{\mbox{\tiny VFF}}(\Delta t)$ dips below this value, then the variational portion of VFF can be terminated. Specifically, the termination condition is $C_{\mbox{\tiny LHST}}^{\mbox{\tiny VFF}}(\Delta t) \leq C_{\text{Threshold}}$, where
\begin{align}
\label{eqn:termination}
    C_{\text{Threshold}} \approx \frac{1}{n} \left(\frac{1}{N}\sqrt{\frac{d+1}{ d}(1-\overline{F}(T))} -\epsilon_{\infty}^{\mbox{\tiny TS}}(\Delta t)\right)^2 \,,
\end{align}
with the approximation holding when $C_{\mbox{\tiny LHST}}^{\mbox{\tiny VFF}}(\Delta t)$ is small. Again, for the exact expression for $ C_{\text{Threshold}} $, see \ref{sec:SimulationErrors}.

\subsection{Implementations}\label{sec:Implementations}

Here we present results simulated classically and on quantum hardware. We refer the reader to Sec.~\ref{sec:Methods} for details about our ansatz and optimization methods.

For our results below, it will be convenient, for a given simulation error tolerance $\delta$, to define {\it fast-forwarding} as the case whenever $R^{\mbox {\tiny FF}}_{\delta} > 1$, where
\begin{equation}
    R^{\mbox {\tiny FF}}_{\delta} = T^{\mbox {\tiny FF}}_{\delta} / T^{\mbox {\tiny Trot}}_{\delta}
\end{equation} 
is the ratio of the simulation time achievable with VFF, $T^{\mbox {\tiny FF}}_{\delta}$, to the simulation time achievable with standard Trotterization, $T^{\mbox {\tiny Trot}}_{\delta}$.  We note that $T^{\mbox {\tiny Trot}}_{\delta}$ is a good empirical measure of the coherence time, since it can account for both decoherence and gate infidelity, and hence the condition $R^{\mbox {\tiny FF}}_{\delta} > 1$ intuitively captures the idea of simulation beyond the coherence time.

\subsubsection{Comparing VFF to Trotterization and Compiled Trotterizations}

Here, we compare the performance of VFF to that of two other simulation strategies. One of these strategies is the standard Trotterization approach depicted in Fig.~\ref{fig:Decomp}(a). Another strategy is to first optimally compile the Trotterization step to a short-depth gate sequence and use this optimal circuit (in place of the Trotterization step) for the approach in Fig.~\ref{fig:Decomp}(a). We refer to this as the QAQC strategy, since one can use the QAQC algorithm~\cite{khatrietal2019} to obtain the optimal compilation of the Trotterization step. With the QAQC strategy, when finding the optimal short-depth compilation we make no assumptions about the structure of the compiled circuit, which is in contrast with Trotterization, where the structure of the circuit is dictated by interactions in the Hamiltonian.

For concreteness, we consider the task of simulating the XY model, defined by the Hamiltonian
\begin{equation}
    H_\textrm{XY} = - \sum_i X^i X^{i+1} + Y^i Y^{i+1} \ ,
\end{equation}
where $X^i$ and $Y^i$ are Pauli operators on qubit $i$. We consider a five-qubit system with open boundary conditions. From the analytical diagonalization of the XY model \cite{lieb1961two}, it follows that the ansatz for the diagonal matrix $D$ can be truncated at the first nontrivial term, as described in the Methods Section (see Eq.~(\ref{eq:a})).

Figure~\ref{fig:Compare} summarizes our results. Panel (a) shows how the $C^{\mbox {\tiny VFF}}_{\mbox {\tiny LHST}}$ cost function is iteratively minimized during the optimization procedure. We selected four approximate diagonalizations corresponding to different cost values denoted by VFF${}_n$, $n=1,\ldots,4$, depicted by colored circles in Panel (a). Note that the colors match those used in other panels. Panel (b) compares these diagonalizations with different quantum simulation methods: Trotterization (dashed line), QAQC-compiled (dashed-dotted line) and VFF at different stages of optimization. We compare simulated time evolution governed by the XY model and observe how the fidelity decays with evolution time. The fidelity is given by $\mathrm{Tr}( \ketbra{\psi(T)}{\psi(T)}\rho(T))$, where $\ket{\psi(T)}$ is the exact evolved state and $\rho(T)$ is the state obtained with a noisy simulator. The initial state $\ket{\psi(0)}$ was chosen randomly such that it has nonvanishing overlap with every eigenstate of the Hamiltonian. The circuits were simulated using a noisy trapped-ion simulator with error model from Ref.~\cite{trout18}. We took error rates, as specified in their Fig.~14 and reduced them by a factor of five for clearer demonstration of VFF's capabilities. 

Results shown in Panel (b) indicate that QAQC performed better than Trotterization, which is expected as QAQC optimizes in a circuit space larger than that given by Trotterization. Both approaches give better results than VFF$_1$ (red curve). This confirms the intuition that at early stages of the VFF optimization (large values of $C^{\mbox {\tiny VFF}}_{\mbox {\tiny LHST}}$), the error of the diagonalization is too big to outperform other methods. As the cost function is decreased, the length of time one can simulate using VFF increases. Indeed, as one can see from the green, blue, and purple curves (which are associated with cost values $\lessapprox 10^{-2}$), VFF dramatically outperforms both Trotterization and QAQC. 

One can see another important feature in Panel (b). For short simulations Trotterization and QAQC are always more accurate than VFF, no matter how accurate the diagonalization is. That is because for small $T$, there are just a few time steps taken by Trotterization and QAQC implementations and the resulting circuits are shorter than VFF circuits implementing $W$, $D$ and $W^\dagger$. This disadvantage rapidly diminishes since VFF circuits do not grow with $T$ and the only error that impacts the fidelity comes from imperfect diagonalization. On the other hand, Trotter and QAQC circuits grow linearly with $T$ and as a result, fidelity is dominated by noise (and not imperfections in the decomposition).

Panel (c) shows how the fast-forwarding factor $R_\delta^\mathrm{FF}$ and the error in the eigenvalue approximation (pertinent if VFF is used for eigenvalue estimation as discussed in Section~\ref{sec:EnergyEstimates}) depend on the cost function $C^{\mbox {\tiny VFF}}_{\mbox {\tiny LHST}}$. The data suggest power-law dependence in both cases. Bringing the cost function down to $10^{-3}$ allows us to reduce the eigenvalue error below $0.1$ and reach a fast-forwarding factor of approximately $30$. Note that for VFF to be more efficient than Trotterization ($R_\delta^\mathrm{FF} > 1$), one has to lower the cost function below approximately $0.04$. For this case, $\delta$ is defined as $1 - \mathrm{Tr}( \ketbra{\psi(T)}{\psi(T)}\rho(T))$ and we considered $\delta=0.2$. Panel (d) presents a more detailed analysis of the eigenvalue error, showing how the error of individual eigenvalues is reduced as the cost function is minimized. 

\subsubsection{Using VFF to Fast Forward
Models Across a Range of Parameters}\label{sec:ImplementationsSimulator}

Here, we show how to use VFF to efficiently find diagonalizations for new models that are nearby in parameter space, from previously diagonalized models.

\bigskip

\textit{Hubbard Model.---}We applied VFF to Trotterized quantum simulation unitaries, $U(\Delta t) \approx e^{-i H_\mathrm{Hub} \Delta t}$, of the Fermi-Hubbard model
\begin{equation}
   H_\mathrm{Hub} = -\tau \sum_{i,j,\sigma} \left( c^\dagger_{i,\sigma} c_{j,\sigma} + c^\dagger_{j,\sigma} c_{i,\sigma} \right) + u \sum_{i=1}^N n_{i,\uparrow} n_{i,\downarrow} \; .
\end{equation}
Here, $c^\dagger_{i,\sigma}$ and $c_{i,\sigma}$ are electron creation and annihilation operators (resp.) for spin $\sigma \in \{ \downarrow, \uparrow \}$ at site $i$ and $n_{i,\sigma} = c^\dagger_{i,\sigma} c_{i,\sigma}$ is the electron number operator. The parameters $\tau$ and $u$ are the hopping strength and on-site interaction (resp.). We studied a two-site, two-qubit Fermi-Hubbard model \cite{linke2017measuring}, which, after translation via the Jordan-Wigner transform, takes the form
\begin{equation}
H_\mathrm{Hub,2} = -\tau(X \otimes I + I \otimes X) + u Z \otimes Z \; .
\end{equation}
We took $\tau = 1$ for our initial diagonalization, then perturbatively increased $u$ from $0$ to $0.1$ in increments of $0.01$. For $U(\Delta t)$, we used a first-order Trotterization of $\exp(-i H_\mathrm{Hub,2} \Delta t)$. We set a threshold for optimization of $10^{-6}$. We used a three-layer ansatz for $W$ and a two-layer ansatz for $D$, which we describe in Section~\ref{sec:Layered}. 

In representative results shown in Fig.~\ref{fig:Hubbard}(a,b), we see that, after an initial optimization taking a number of iteration steps, VFF reached the optimization threshold. Then, as we perturbed away from $u = 0$, VFF rapidly found new parameters that diagonalized $\exp(-i H_\mathrm{Hub,2} \Delta t)$ to below the cost threshold. For all approximate diagonalizations, for an error tolerance of $\delta = 10^{-2}$, the simulation error remained below this tolerance for $T = 30 \Delta t$. The diagonalization used $9$ single-qubit gates and $7$ two-qubit gates. The Trotterization used $2$ single-qubit gates and $1$ two-qubit gate. Thus, the fast-forwarded simulations used $9$ single-qubit layers and $7$ two-qubit gates, but the equivalent Trotterized simulations used $60$ single-qubit gates and $30$ two-qubit gates. Thus, VFF gave significant depth compression versus the Trotterized simulations, particularly with respect to entangling gates.

\begin{figure*}[!htbp]
\centering
\includegraphics[width=0.98\textwidth]{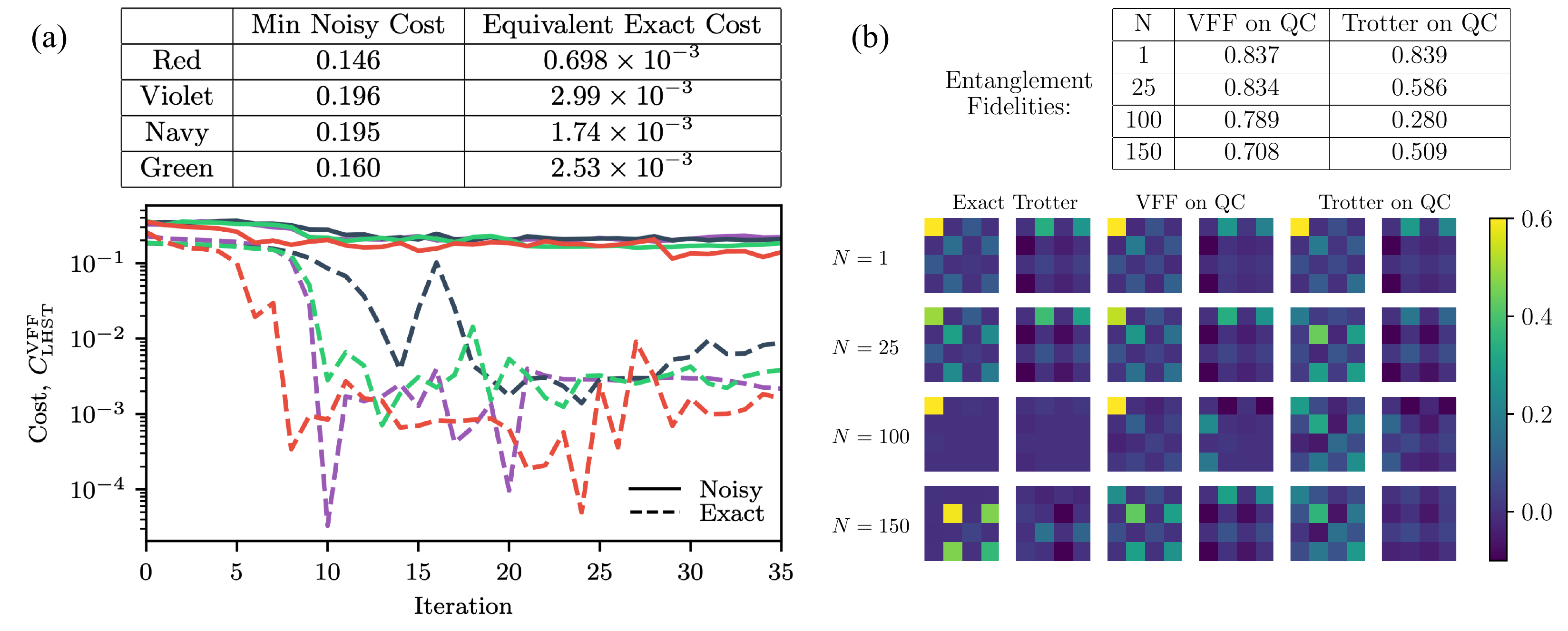}
\caption{\label{fig:Rigetti} VFF on quantum hardware. (a) Training results for single-qubit VFF implemented on the Rigetti Aspen-4 quantum computer. Here, the quantum circuit acted on two qubits, one with a random single-qubit unitary, $U$, and the second with the diagonal ansatz, $V = W D W^\dagger$. Optimization was performed using gradient descent of the VFF cost function. Results from four optimizations are shown. The plot shows the cost function evaluated on the QC (solid line) and the true (noiseless) cost function evaluated classically (dashed line) using the parameters found on the Rigetti QC via VFF. The table in (a) provides the optimal noisy cost values from the Rigetti QC and the equivalent true cost value for the given set of optimized parameters. (b) Process tomography for single-qubit VFF implemented on the Rigetti Aspen-4 quantum computer. Real (left) and imaginary (right) parts of the exact, classically computed process matrix of a first-order Trotterized quantum simulation (Exact Trotter) compared with a quantum simulation using an optimal diagonalization from the VFF shown in (a) (VFF on QC) and the first-order Trotterization (Trotter on QC), both computed on the Rigetti QC. The number of timesteps for the simulation are shown to the left. Note that for the VFF simulation, we used the optimization angles corresponding to the best cost from the {\it noisy} cost function, i.e., what was actually measured on the QC. To quantify the accuracy of the fast-forwarded simulation, we include a table in (b) containing the entanglement fidelity~\cite{EntFidelityNielsen} between the exact unitary and either the noisy process implemented by VFF or Trotterization respectively on the Rigetti QC.}
\end{figure*}

\bigskip

\textit{Heisenberg Model.---}Next, we applied VFF to the Heisenberg model,
\begin{equation} \label{eq:heisHam}
   H_\mathrm{Heis} = \sum_i J_z Z^i Z^{i+1} + J_x X^i X^{i+1} + J_y Y^i Y^{i+1} + h Z^i \; ,
\end{equation}
where $X^j$, $Y^j$, and $Z^j$ are Pauli spin matrices acting on qubit $j$, and $h$, $J_x$, $J_y$, and $J_z$ are parameters.

Here, we took $h = 1.0$ and investigated the model acting on three qubits (whose Hamiltonian we denote $H_\mathrm{Heis,3}$). We used a first-order Trotterization of $\exp(-i H_\mathrm{Heis,3} \Delta t)$. We set an optimization threshold of $10^{-6}$ and used a ten-layer ansatz for $W$ and a two-layer ansatz for $D$. From $J_z = 1.0$ (a non-interacting Hamiltonian) we increased $J_z$ to $5.0$ in increments of $1.0$. For these parameter values, $H_\mathrm{Heis}$ is an anti-ferromagnetic classical Ising model. 

Next, we kept $h = 1.0$ and $J_z = 5.0$ fixed and increased $J_x = J_y$ from $0.0$ to $8.0$ in increments of $2.0$. When $J_x = J_y$, these are often called XXZ Heisenberg models. 

Finally, we kept $h = 1.0$, $J_z = 5.0$, $J_x = 8.0$ and varied $J_y$ from $0.0$ to $10.0$ in increments of $1.0$ (XYZ Heisenberg models). 

As may be seen in the representative results plotted in Fig.~\ref{fig:Hubbard}(c,d), VFF rapidly found new diagonalizations $WDW^\dagger \approx \exp(-i H_\mathrm{Heis,3} \Delta t)$ for all models considered. We performed additional searches for diagonalizations of ferromagnetic models ($J_z$, $J_x$, and $J_y < 0$) with similar results. For all approximate diagonalizations, the simulation error remained below an error tolerance of $\delta = 10^{-2}$, up to $T \approx 100 \Delta t$. For this simulation time, each diagonalization used $40$ two-qubit gates and $71$ single-qubit gates ($111$ total), whereas each Trotterization used $1200$ two-qubit gates and $2500$ single-qubit gates ($3700$ total). 

\subsubsection{VFF Implemented on Quantum Hardware}\label{sec:Rigetti}

We implemented VFF on $1 + 1$ qubits (i.e. diagonalizing a random single-qubit unitary) on the Rigetti Aspen-4 quantum computer (Fig.~\ref{fig:Rigetti}). Here we considered the first-order Trotterization of the Hamiltonian $H = \alpha_x \sigma_x + \alpha_y \sigma_y + \alpha_z \sigma_z$, where $\alpha$ was a randomly chosen unit vector, at the time $\Delta t = 0.5$. We used $W =  R_z(\theta_z) R_x(\theta_x)$ and $D = R_z(\gamma_z)$. The VFF cost function, as evaluated on the QC with $n_\mathrm{samp} = 10^4$, was optimized to $C^{\mbox {\tiny VFF}}_{\mbox {\tiny LHST}}(\Delta t) \approx 10^{-1}$. 

With this system, we investigated how well VFF performed by classically computing the true, noiseless, cost for the parameters found on the Rigetti QC. This true cost converged to two orders of magnitude below the QC-evaluated cost, demonstrating significant robustness of VFF to the noise on the Rigetti QC.

We next simulated single qubit evolution on the QC by 1) iterating the original Trotterization, $U(\Delta t)^\mathrm{N}$, and 2) using the VFF diagonalization (\ref{VFF}). We then used process tomography to compare the resultant noisy process resulting from the Trotterization and the process resulting from VFF to the exact process $U(\Delta t)^\mathrm{N}$ calculated classically.

In this single qubit case, the Trotterized simulation unitary could have been compiled to a circuit with many fewer gates; however, this would not be true for higher dimensional unitaries and for this reason we did not compile the iterated gate sequence here. 

In Fig.~\ref{fig:Rigetti}(b), we show that VFF performed much better than the iterated Trotterization, giving a high fidelity simulation. In these results, the entanglement fidelity between the process implemented using VFF and the exact process remained high until at least $N_\mathrm{VFF} = 150$ and never reached a value below $0.7$. On the other hand, the fidelity of the iterated Trotterization approach was already $0.586$ by $N=25$.

These results demonstrate that VFF on current quantum computers can allow for simulations beyond the coherence time. 
For example, taking an entanglement infidelity of 0.3 as our error tolerance $\delta$, it follows from the table in Fig.~\ref{fig:Rigetti}(b) that we obtained a fast-forwarding beyond the coherence time of at least $R^{\mbox {\tiny FF}}_{\delta} = 6$.

\subsubsection{Estimating Energy Eigenvalues}\label{sec:EnergyEstimates}

We primarily foresee VFF being used to study the long time evolution of the observables of a system. But one may also use VFF to reduce the gate complexity of eigenvalue estimation algorithms such as Quantum Phase Estimation (QPE)~\cite{KitaevQPE} or time series analyses~\cite{Somma2002TS, Somma2019TS}. Such algorithms require simulating a Hamiltonian up to time $T = \OC\left(\frac{1}{\sigma}\right)$ to obtain eigenvalue estimates of accuracy $\sigma$. Due to the constant depth circuits produced by VFF, we can therefore reduce the number of gates required for these algorithms by a factor of $\OC\left(\frac{1}{\sigma}\right)$, increasing the viability of eigenvalue estimation on noisy quantum computers. 

The eigenvalues of the Hamiltonian simulated via VFF are not directly accessible from the diagonal unitary $D$ since they are encoded in a set of Pauli operators. However, these can be extracted using the time series analysis in \cite{Somma2019TS}. This method does not require large ancillary systems nor large numbers of controlled-unitary operations and thus is a promising avenue for eigenvalue estimation in the NISQ era.

To demonstrate the practical utility of VFF for eigenvalue estimation we numerically computed the spectrum of a two-site Hubbard model in Fig.~\ref{fig:Eigenenergy}(a-c) and in Fig.~\ref{fig:Eigenenergy}(d), we show eigenenergy estimates for a 5-qubit XY model. Specifically, we used a one-clean-qubit (DQC-1) quantum circuit to discretely sample the function
\begin{equation}
    g(t) = \mathrm{Tr}\left( \dya{\psi} \; e^{-iDt} \right) = \sum_j e^{-i \lambda_j t} \,,
\end{equation}
where $| \psi \rangle := \frac{1}{\sqrt{2}}(\ket{0} + \ket{1})^{\otimes n}$, and then used classical time series analysis to estimate the eigenvalues. This is achieved by computing each spectral estimate $S(\lambda)$ with respect to a discrete number of values for time variable, $t_j = \{ 0, \dots, t_{\mbox{\tiny max}}  \}$ in increments of $0.2 \Delta t$. A higher number of discrete points results in a better resolution of $S(\lambda)$. The signal processing uncertainty principle constrains the spectral widths (variance in the estimate of $\lambda$) to obey $\sigma_{\lambda_j} t_{\mbox{\tiny max}} \ge c$, where $c$ is a constant of order $1$. In Fig.~\ref{fig:Eigenenergy}(a-c), we show three examples with successively better optimization and, hence, longer integration times, $t_\mathrm{max}$. We plot eigenenergy estimates of diagonalized two-site Hubbard models with parameters $\tau = 1$ and $u \in \{ 0.0, 0.2, \dots, 1.0 \}$ ranging from weakly to strongly coupled models.

\begin{figure}[t]
\centering
\includegraphics[width=0.48\textwidth]{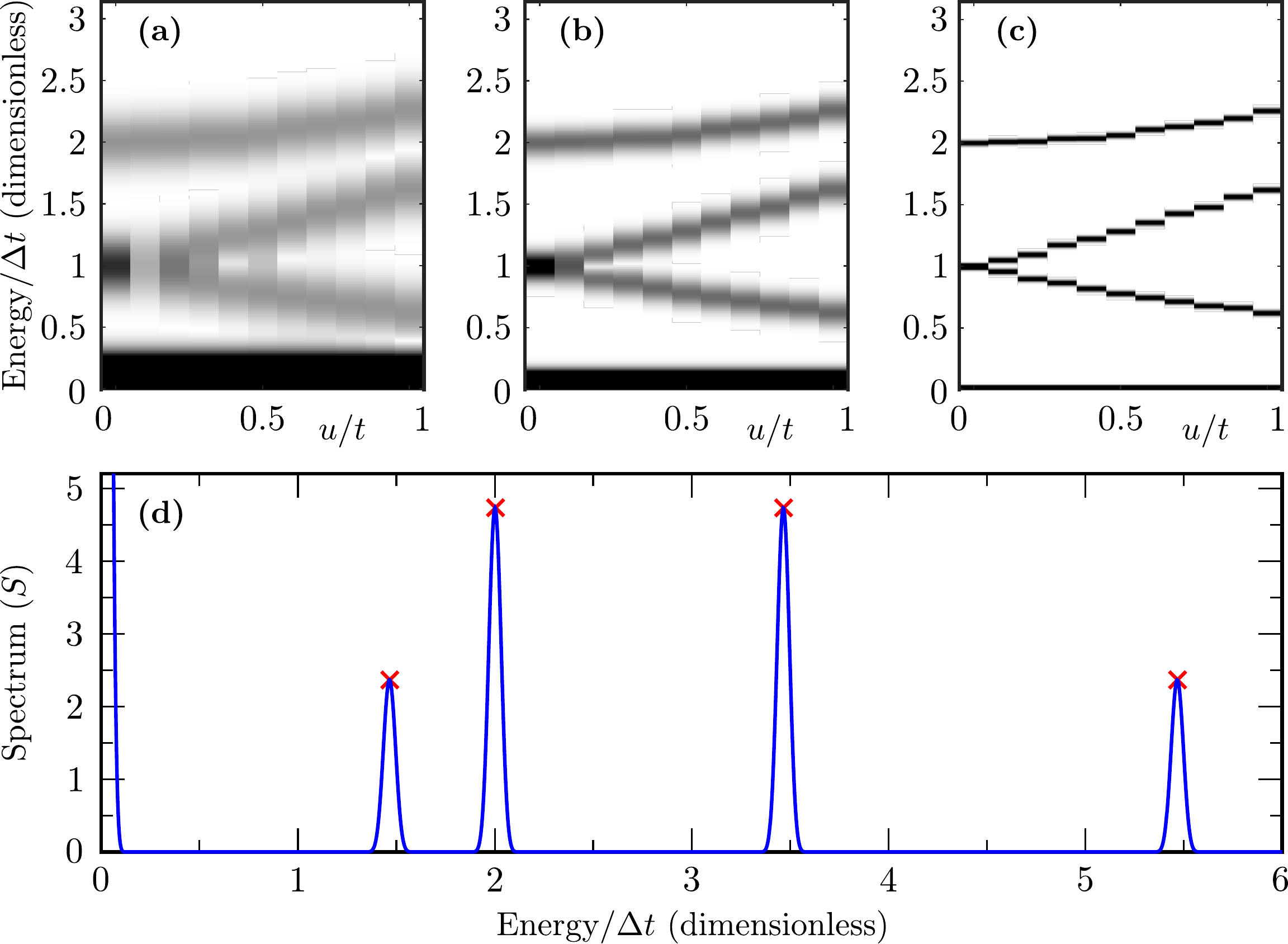}
\caption{\label{fig:Eigenenergy} Estimating energy eigenvalues using VFF. (a,b,c) $S(\lambda)$ estimates from a VFF diagonalization of a two-site Hubbard model with $u/\tau = \{0.0, \dots, 1.0\}$. Only positive eigenenergies are shown. For each of (a,b,c) we chose $t_{\mbox{\tiny max}}$ to be $T^{\mbox {\tiny FF}}_{10^{-2}}$, the simulation time achievable with a cost less than $10^{-2}$, with $t_{\mbox{\tiny max}} = T^{\mbox {\tiny FF}}_{10^{-2}} = 500 \Delta t, 1000 \Delta t, 5000 \Delta t  $ in (a), (b), (c) respectively. The resolution of $S(\lambda)$ scales inversely with $T^{\mbox {\tiny FF}}_{10^{-2}}$ causing the width of the spectral peaks to get successively narrower as $t_{\mbox{\tiny max}} = T^{\mbox {\tiny FF}}_{10^{-2}}$ is increased. (d) $S(\lambda)$ estimate from a VFF diagonalization of a 5-qubit Heisenberg XY model. The different spectral heights are due to degenerate eigenvalues (e.g. the multiplicity of the peak at an energy of $2$ is twice that at $1.4$)}
\end{figure}

The time series analysis extracts an estimate for the spectrum of energies corresponding to $V$, the approximate unitary given by VFF of the target unitary $U$, up to a global phase. The Hoffman-Wielandt theorem \cite{hoffman1953} gives a bound on the total variation distance between spectra of $U$ and $V$, in terms of the 2-norm which in turn is directly related to the VFF cost function. For the concrete bounds we refer to \ref{ap:Eigen}. This illustrates that the estimated spectral differences resulting from classical time series analysis give better approximations to the energy differences of the target Hamiltonian when decreasing our VFF cost function. In Fig.~\ref{fig:Compare}(c,d) we provide additional numerical analysis supporting this feature.

\section{Discussion}\label{sec:Discussion}

We presented a new variational method for quantum simulation called Variational Fast Forwarding (VFF). Our results showed that, once a diagonalization is in hand, one could form an approximate fast forwarding of the simulation that allowed for quantum simulations beyond the coherence time. We compared VFF simulation fidelities for a range of optimization errors with Trotterized and compiled-Trotterized simulations and showed that, as long as the VFF optimization error was sufficiently small, VFF could indeed fast-forward quantum simulations. For the particular models, ans\"atze, and thresholds that we studied, we were able to fast forward simulations by factors of approximately $30$ (Hubbard) and $\sim 80$ (Heisenberg) simulation timesteps. In addition, a fast-forwarding of a factor of at least $6$, relative to a Trotterization approach, was found experimentally on Rigetti's quantum hardware. We also explored the use of VFF for simplifying eigenenergy estimates and showed that the variance of eigenenergy estimates is reduced commensurately with the cost function. Essentially, the more accurate the diagonalization step of VFF is (i.e., the lower the cost function value), the longer is the achievable fast-forwarding simulation time and the better the eigenenergy estimate.

A crucial feature of VFF is the operational meaning of its cost function as a bound on average-case simulation error. Hence, any reduction in the cost results in a tighter bound on the simulation error. We used this feature to define a termination condition for the variational portion of VFF, such that once the cost is below a particular value, one can guarantee that the simulation error will be below a desired threshold. This is arguably the most important feature that distinguishes VFF from prior work on Subspace Variational Quantum Simulation (SVQS)~\cite{HeyaEtAl2019}, whose cost function does not have an obvious meaning in terms of simulation error. In addition, since VFF is not targeting a low-energy subspace, it is capable of simulating systems at moderate to high-temperature or more dramatic dynamics such as quenches. The tradeoff is that the diagonalization step of VFF can be more difficult than that of SVQS, since one is diagonalizing over the entire space rather than a subspace. This tradeoff will be important to study in future work.

In the NISQ era, the minimum value of the VFF cost function that can be achieved will be limited by quantum hardware noise. On the one hand, this will result in loose bounds on the simulation error obtained from~\eqref{eqn:approxLowerBoundF}. On the other hand, we have seen from our implementation of VFF on Rigetti's quantum hardware that the true (noiseless) cost is often orders of magnitude lower than the noisy cost, implying that we learned the correct optimal parameters despite the noise. This noise resilience is analogous to analytical and numerical results recently reported in \cite{SharmaEtAl2019}. Hence, an important direction of future research would be to tighten our bound \eqref{eqn:approxLowerBoundF} for specific noise models, which would allow for tight simulation error bounds in the presence of noise.

Finally, a possible limitation of the scalability of VFF is the No Fast-Forwarding Theorem, which is stated in a variety of forms~\cite{atia2017fast,BerryEtAl2007,Childs:2010:LSN:2011373.2011380}, but basically says that there exist some families of Hamiltonians for which asymptotically the number of gates needed for quantum simulation must grow roughly in proportion to the simulation time. Thus VFF may not work for large scale and/or long time simulations of these Hamiltonians, perhaps because the circuit depth needed to achieve an accurate diagonalization will be long or perhaps because the cost landscape will be difficult to optimize. 
Nonetheless, there are many physically interesting Hamiltonians that are fast-forwardable or close to (i.e., perturbations of) models that are known to be fast-forwardable.
Moreover, fast-forwardable Hamiltonians can generate highly non-classical behavior~\cite{novo2019quantum}.
Hence, future work needs to explore the class of Hamiltonians that are approximately fast-forwardable.

\section{Methods}\label{sec:Methods}

\subsection{Ansatz}\label{sec:Layered}

As with many variational quantum-classical algorithms, it is natural to employ a layered gate structure for $W(\thv)$ and $D(\gav, \Delta t)$, with the number of layers being a refinement parameter. 

\subsubsection{Ansatz for $D$}

Let us first consider an ansatz for $D$. The problem of constructing quantum circuits for diagonal unitaries, $D$, is equivalent to finding a Walsh series approximation~\cite{Welch2014}
\begin{equation} \label{eq:walsh}
  D = e^{i G} = \prod_{j=0}^{2^q-1} e^{i \gamma_j \bigotimes_{k=1}^n (Z_k)^{j_k}} \; ,
\end{equation}
where $q = n$, $G$ and $Z_k$ are diagonal operators with the Pauli operator $Z_k$ acting on the $k$-th qubit, and $j_k$ is the $k$-th bit in a bitstring $j$. Efficient quantum circuits for minimum depth approximations of $D$ may be obtained by resampling the function on the diagonal of $G$ at sequencies lower than a fixed threshold, with $q = k$ and $k \leq n$. The resampled diagonal takes the same form as \eqref{eq:walsh} but with $q = k$. The error after resampling is $\epsilon_k \leq \mathrm{sup}_x |G'(x)|/2^k$, where we have introduced a coordinate along the diagonal, $x$. While we do not know $G$, we can assume a particular ansatz for terms to include in the expansion.

In all of our implementations, we use a re-ordering of terms in Eq.~(\ref{eq:walsh}). Namely, we take
\begin{equation} \label{eq:a}
    D = \prod_{m=0}^n \prod_{j \in S_m} e^{i  \gamma_j 
    \bigotimes_{k=1}^n (Z_k)^{j_k}} \ ,
\end{equation}
where $S_m$ is a set of all indices $j$ such that $\sum_{k=1}^n j_k = m$. Note that the $l$-local terms, $\bigotimes_{k=1}^n (Z_k)^{j_k}$, $\sum_{k=1}^n j_k = l$, in Eq.~(\ref{eq:a}) are organized in increasing order. We truncate the above product to a small number (up to $l=2$) of initial $l$-local terms. The accuracy of the approximation is controlled by truncating the expansion in Eq.~(\ref{eq:a}). The above expansion may be more suited than Eq.~(\ref{eq:walsh}) for quantum many-body Hamiltonians. For instance, it is known that the quantum Ising model in a transverse field can be diagonalized exactly by keeping only $1$-local terms.

\subsubsection{Ansatz for $W$}

Let us now consider an ansatz for $W(\thv)$. With the Baker-Campbell-Hausdorff formula we may generate any eigenvector unitary, $W(\thv)$, by appropriately interleaving non-commuting unitaries \cite{LloydQuantumSim,khatrietal2019}. In general, this requires order $d^2$ parameterized operations. Here, we briefly discuss two approaches to make its generation tractable.

The first approach is to use a fixed, layered ansatz for $W(\thv)$. By alternating sets of single- and two-qubit unitaries, we construct a polynomial number of non-commuting layers capable of generating a rich set of parameterized unitaries. Translational invariance of the system Hamiltonian may be incorporated into the ansatz for $W(\thv)$. In this case, all gates in a given layer may be chosen to be the same. As a result, the number of variational parameters is reduced by a factor of $n$. 

Another approach is to employ a randomized ansatz, in which parameterized gates are randomly placed.
This approach may be more suitable for irregular Hamiltonians $H$, where the optimal form of $W(\thv)$ is not easily deducible from $H$. The randomized approach may potentially find a shorter $W(\thv)$ that contains fewer gates, which is beneficial for near-term applications. Ref.~\cite{larose2018} discusses further details of both methods.

\subsubsection{Growing the Ansatz and Parameter Initialization}

We use the method of growing the ansatz in order to mitigate the problem of getting trapped in local minima during the optimization~\cite{larose2018, grant2019initialization}. This technique can be used with both ans\"atze mentioned above. The optimization is initiated with a shallow circuit containing only a few variational parameters. After a local minimum is found, we add a resolution of the identity to the ansatz for $W(\thv)$. This takes a form of a layer of unitaries (for a layered ansatz) or a smaller block of parameterized gates (for a randomized ansatz) that evaluates to the identity. Adding such structures to $W(\thv)$ does not change the value of the cost function but it increases the number of variational parameters. In the enlarged space, local minima encountered in previous steps may be turned into saddle points and the cost function may be further minimized towards the global minimum. The technique of systematically growing the ansatz to improve the quality of the result and mitigate the problem of local minima is described in detail in \cite{larose2018}.

In order to approach the issue of initializing the parameters $\thv$ and $\gav$, we often use a perturbative method~\cite{McCleanPerturb,GarciaSaezPerturb} in which we pre-train these parameters for a slightly different Hamiltonian. Namely, we begin a VFF search for a unitary diagonalization with a known short-depth, readily diagonalizable, unitary. We then modify the Hamiltonian by adding successively perturbed terms in an attempt to guide the previously learned diagonalization from known initial parameters toward an unknown diagonalization of interest. 

\subsubsection{Ansatz for Implementations}

General ansatz considerations were discussed above, and now we discuss the specific ans\"atze used in our implementations. For our implementations, $W$ consists of successive layers, each formed of three sub-layers: (i) an initial sub-layer of single-qubit gates, (ii) a second sub-layer of entangling two-qubit gates acting on neighboring even-odd qubit pairs, and (iii) a third sub-layer of two-qubit gates acting on odd-even qubit pairs. The two-qubit gates are typically CNOTs, but equivalently we have used $\mathrm{ZZ}(\theta) = \mathrm{CNOT} (I\otimes R_z(\theta)) \mathrm{CNOT}$ or $\mathrm{XX}(\theta)$ gates. The layers are appended successively always with a final layer of single-qubit gates.

In addition, our implementations use a set of layers consisting of various commuting operators for $D$. For the first layer we use a set of single-qubit $Z$-rotations, $R_z(\gamma)$, acting on all qubits. The second layer is a set of two-qubit $ZZ(\gamma)$ gates acting on all pairs of qubits. The third layer would be a set of three-qubit gates $Z\otimes Z\otimes Z(\gamma)$ acting on all triplets of qubits. However, for the threshold used, we did not need a third layer for the results in Sec.~\ref{sec:Implementations}.

\subsection{Optimization via Gradient Descent}\label{sec:gd}

Gradient-based approaches can improve convergence of variational quantum-classical algorithms~\cite{harrow2019low}, and the optimizer performance can be further enhanced by judiciously adapting the shot noise for each partial derivative~\cite{kubler2019adaptive}. Furthermore, the same quantum circuit used for cost estimation can be used for gradient estimation~\cite{mitarai2018quantum}.  Therefore, we recommend a gradient-based approach for VFF, in what follows.

With the ansatz in \eqref{eqn:Vansatz}, we denote the VFF cost function as $C_{\mbox{\tiny LHST}}^{\mbox{\tiny VFF}} :=  C_{\mbox{\tiny LHST}}(U, W D W^\dagger)$. The partial derivative of this cost function with respect to $\theta_k$, a parameter of the eigenvector operator $W(\thv)$, is
\begin{equation} \label{eq:CostEigenvectorGradient}
    \begin{split}
        \frac{\partial C_{\mbox{\tiny LHST}}^{\mbox{\tiny VFF}}}{\partial \theta_k} = \frac{1}{2} \Big( & C_{\mbox{\tiny LHST}}(U, W_+^k D W^\dagger) \\
        - \ &  C_{\mbox{\tiny LHST}}(U, W_-^k D W^\dagger) \\
        + \ & C_{\mbox{\tiny LHST}}(U, W D (W_+^k)^\dagger) \\
        - \ & C_{\mbox{\tiny LHST}}(U, W D (W_{-}^k)^{\dagger}) \Big) \ .
    \end{split}
\end{equation}
The operator $W_+^k$ ($W_-^k$) is generated from the original eigenvector operator $W(\thv)$ by the addition of an extra $\frac{\pi}{2}$ ($-\frac{\pi}{2}$) rotation about a given parameter's rotation axis:
\begin{equation}
    W_{\pm}^k := W\left( \thv_{\pm}^k \right) \ \ \text{with}  \ \ (\theta_{\pm}^k)_i := \theta_i \pm \frac{\pi}{2} \delta_{i,k} \; .
\end{equation}
Similarly, the partial derivative with respect to $\gamma_\ell$, a parameter of the diagonal operator $D(\gav)$, is
\begin{equation} \label{eq:CostDiagonalGradient}
    \begin{split}
        \frac{\partial C_{\mbox{\tiny LHST}}^{\mbox{\tiny VFF}}}{\partial \gamma_\ell}   =  \frac{1}{2} \Big( & C_{\mbox{\tiny LHST}}\left( U, W  D_+^{\ell} W^\dagger \right)  \\
     - \ & C_{\mbox{\tiny LHST}}\left( U, W  D_-^{\ell} W^\dagger \right) \Big)
    \end{split}
\end{equation}
with 
\begin{equation}
    D_{\pm}^{\ell} := D\left( \gav_{\pm}^{\ell} \right) \ \ \text{with}  \ \ (\gamma_{\pm}^{\ell})_i := \gamma_i \pm \frac{\pi}{2} \delta_{i,\ell} \; .
\end{equation}
Equation~\eqref{eq:CostDiagonalGradient} is derived in~\cite{khatrietal2019} and we derive Eq.~\eqref{eq:CostEigenvectorGradient} in \ref{ap:GradientCalc}.

Using~\eqref{eq:CostEigenvectorGradient} and~\eqref{eq:CostDiagonalGradient}, we can evaluate the gradient of $C_{\mbox{\tiny LHST}}^{\mbox{\tiny VFF}}$ directly and use a simple gradient descent iteration
\begin{eqnarray}
  \theta^{(t+1)}_k & = & \theta^{(t)}_k - \eta \frac{\partial C_{\mbox{\tiny LHST}}^{\mbox{\tiny VFF}}}{\partial \theta_k} \\
  \gamma^{(t+1)}_\ell & = & \gamma^{(t)}_\ell - \eta \frac{\partial C_{\mbox{\tiny LHST}}^{\mbox{\tiny VFF}}}{\partial \gamma_\ell} \; , 
\end{eqnarray}
to minimize $C_{\mbox{\tiny LHST}}^{\mbox{\tiny VFF}}$.

\subsection{Data availability} The authors declare that the main data supporting the findings of this study are available within the article and its Supplementary
Information files. Extra data sets are available upon request.

\section{Acknowledgments}

We thank Rolando Somma and Sumeet Khatri for helpful discussions. We thank Rigetti for providing access to its quantum computers. The views expressed in this paper are those of the authors and do not reflect those of Rigetti. CC, ZH, and JI acknowledge support from the U.S. Department of Energy (DOE) through a quantum computing program sponsored by the LANL Information Science \& Technology Institute. CC acknowledges support from the EPSRC National Quantum Technology Hub in Networked Quantum Information Technologies. ZH acknowledges support from the EPSRC Centre for Doctoral Training in Controlled Quantum Dynamics. LC was supported initially by the U.S. DOE through the J. Robert Oppenheimer fellowship and subsequently by the DOE, Office of Science, Basic Energy Sciences, Materials Sciences and Engineering Division, Condensed Matter Theory Program. PJC and AS acknowledge initial support from the DOE ASC Beyond Moore's Law program and subsequent support from LANL's Laboratory Directed Research and Development (LDRD) program. 

\medskip

\textbf{Competing financial interests:} The authors declare no competing financial interests.

\section{Author contributions}
The project was conceived by LC, PJC, and AS. The VFF algorithm was formulated by CC, ZH, JI, LC, PJC, and AS. The manuscript was written by CC, ZH, PJC, and AS. For the analytical results, CC and PJC performed the error analysis, PJC derived the termination condition, and ZH and AS derived the gradient formulas. For the numerical results, LC and AS performed the simulator implementations, while ZH performed the quantum hardware implementations.

\vfill
%


\clearpage
\onecolumngrid


\setcounter{equation}{0}
\setcounter{figure}{0}
\setcounter{table}{0}
\setcounter{page}{1}
\setcounter{section}{0}
\makeatletter
\renewcommand{\thesection}{\color{black}{Supplementary Note}~\arabic{section}}
\renewcommand{\theequation}{S\arabic{equation}}
\renewcommand{\thefigure}{S\arabic{figure}}
\renewcommand{\bibnumfmt}[1]{[S#1]}

\begin{center}
\Large{ Supplementary Material for \\ ``Variational Fast Forwarding for Quantum Simulation Beyond the Coherence Time''}
\end{center}


\section{Cost Function}\label{ap:Cost}

Here we elaborate on our cost function. As noted in Sec.~\ref{sec:cost}, our proposed cost function is the $C_{\mbox {\tiny LHST}}$ function introduced in Ref.~\cite{khatrietal2019}, defined by 
\begin{equation}
    C_{\mbox {\tiny LHST}}(U, V)= 1- \frac{1}{n}\sum_{j=1}^n F_e^{(j)}\,.
\end{equation}

Let us now precisely define the entanglement fidelities $F_e^{(j)}$. Consider a $2n$-qubit system composed of the $n$-qubit subsystems $A$ and $B$. Let $A_j$ ($B_j$) denote the $j$-th qubit of system $A$ ($B$). Let $\ket{\Phi^+} = (\ket{00} + \ket{11})/\sqrt{2}$ denote the standard 2-qubit Bell state. Then we can write $F_e^{(j)}$ as
\begin{equation}\label{eq-LHST_prob}
 F_e^{(j)} :=  \Tr\left(\dya{\Phi^+}_{A_jB_j}(\mathcal{E}_{j}\otimes\mathcal{I}_{B_j})(\dya{\Phi^+}_{A_jB_j})\right).
\end{equation}
Here, $\EC_j$ is a quantum channel that acts on qubit $A_j$ as follows. For an arbitrary state $\rho_{A_j}$,
\begin{equation}\label{eq-HST_local_channel}
    \EC_j(\rho_{A_j}) = \Tr_{\overline{A}_j}\left(U V\ad \left(\rho_{A_j}\otimes\frac{\id_{\overline{A}_j}}{2^{n-1}}\right)V U\ad \right),
\end{equation}
where $\overline{A}_j$ is the set of all qubits in $A$ except for $A_j$.

The fact that $C_{\mbox {\tiny LHST}}$ is a faithful cost function was shown in \cite{khatrietal2019} by relating $C_{\mbox {\tiny LHST}}$ to another cost function whose properties are more transparent. Namely, consider the function
    \begin{align}\label{eq:HSTdefintion}
    C_{\mbox {\tiny HST}}(U,V) = 1-\frac{1}{d^2}|\Tr(U V\ad)|^2.
    \end{align}
Since $\Tr(U V\ad)$ is the Hilbert-Schmidt inner product, it is clear that $C_{\mbox {\tiny HST}}(U,V) = 0$ if and only if $V = U$ (up to global phase). Reference~\cite{khatrietal2019} then proved the following relation: 
\begin{equation}\label{eq:LHSTvsHSTbound}
    C_{\mbox {\tiny LHST}}(U,V) \leq C_{\mbox {\tiny HST}}(U,V) \leq n \, C_{\mbox {\tiny LHST}}(U,V) \,.
\end{equation}
This implies that $C_{\mbox {\tiny LHST}}$ vanishes under precisely the same conditions as $C_{\mbox {\tiny HST}}$, and hence that $C_{\mbox {\tiny LHST}}$ is faithful.

While the $C_{\mbox {\tiny HST}}$ function has direct operational meaning in terms of the inner product between $U$ and $V$, we propose to use $C_{\mbox {\tiny LHST}}$ instead of $C_{\mbox {\tiny HST}}$ for the following reason. In Ref.~\cite{khatrietal2019}, it was argued that there are simple examples (e.g., when $U$ and $V$ are tensor-product unitaries) for which the gradient of $C_{\mbox {\tiny HST}}$ vanishes exponentially with $n$, while the gradient of $C_{\mbox {\tiny LHST}}$ is independent of $n$. This implies that $C_{\mbox {\tiny LHST}}$ is easier to train than $C_{\mbox {\tiny HST}}$ for large $n$, and indeed Ref.~\cite{khatrietal2019} confirmed this with numerical implementations for increasing values of $n$. Hence $C_{\mbox {\tiny LHST}}$ has better scaling properties, while it also inherits the operational meaning of $C_{\mbox {\tiny HST}}$ via the relation in \eqref{eq:LHSTvsHSTbound}.

\section{Simulation errors}\label{sec:SimulationErrors}

\subsection{Linear scaling in $N$}

Here we provide a proof of \eqref{eqnLemmaNsteps}, which is restated in the following lemma.
\begin{lemma}\label{lemma1}
	Suppose $U_1$ and $U_2$ are two unitary matrices, then for any positive integer $N$ we have:
	\begin{equation}\label{eq:unitaryproductbound}
	||U_1^{N} - U_2^{N}||_{p}\leq N ||U_1 - U_2||_p \; 
	\end{equation}
    where $|| ... ||_p$ denotes the Schatten $p$-norm. 
\end{lemma}
\begin{proof}
	We expand the norm of the difference of products by adding and subtracting convenient terms so that
	\begin{align}
	 	\| U_1^N-U_2^N \|_{p} = \| U_1^{N} - U_1^{N-1}U_2 + U_1^{N-1}U_2  - U_1^{N-2}U_2^2 ... + U_1 U_2^{N-1} - U_2^N \|_{p}   \; . 
	\end{align}
 	From the triangle inequality it follows that:
	\begin{align}
	\| U_1^N-U_2^N \|_{p} \leq \,  \|U_1^{N-1} (U_1-U_2)\|_p \, +  \| U_1^{N-2}(U_1-U_2)U_2 \|_p \,  ... + \| (U_1 - U_2)U_2^{N-1} \|_p \; .
	\end{align}
	 There are a total of $N$ terms in the above summation and as the Schatten norms are unitarily invariant, meaning that $\| U A V \| = \| A \|$ for any unitary matrices $U$ and $V$, each of these $N$ terms is equal to $\| U_1 - U_2 \|$. Thus we obtain the required result.
\end{proof}

\subsection{Scaling of cost function with $N$}\label{ap:CertBounds}

Here we reformulate Lemma~\ref{lemma1} in terms of the VFF cost function. Since the latter can be efficiently estimated on a quantum computer, one can view this reformulation as a certifiable version of Lemma~\ref{lemma1}. This reformulation is derived by specializing Lemma~\ref{lemma1} to the case of $p=2$, i.e. the Hilbert-Schmidt norm.

We first remark that for our purposes the phases of the simulated unitaries are always global phases and therefore unphysical. (This is true since we intend for $V$ to be implemented directly on an $n$-qubit system rather than on a subsytem as a subcomponent of a larger simulation.) We therefore introduce the phase-independent quantity $\tilde{\epsilon}_2(U_1, U_2)$,
\begin{align}
\label{eqnTriangleInequality4}
\tilde{\epsilon}_2(U_1, U_2):= \min_{\phi} \| U_1 - \exp(i \phi) U_2 \|_{2} \,, 
\end{align}
which depends only on the Hilbert-Schmidt inner product, 
\begin{equation}
   \begin{aligned}
   \tilde{\epsilon}_2(U_1, U_2) &= \min_{\phi} \sqrt{2d - 2 \Re(\Tr( U_1^\dagger  e^{i\phi} U_2))} \\
   &= \sqrt{2d - 2 |\Tr (U_1^\dagger U_2 )|} \,.
\end{aligned} 
\end{equation}
Hence this quantity can be related to $C_{\mbox{\tiny HST}}$ defined in \eqref{eq:HSTdefintion} by
\begin{equation}
   \begin{aligned}
   \tilde{\epsilon}_2(U_1, U_2) = \sqrt{2d \left(1 - \sqrt{1-C_{\mbox{\tiny HST}}(U_1, U_2)}\right)} \; .
\end{aligned}\label{eq:POTQtoHST} 
\end{equation}

Now, let us specialize Lemma~\ref{lemma1} to $p=2$ and minimize over all global phase factors applied to the unitary $U_2$. Using \eqref{eq:POTQtoHST}, this results in:
\begin{equation}\label{eq:fastforwardCHST}
    1-\sqrt{1-C_{\mbox{\tiny HST}}(U_1^N, U_2^N)} \leq N^2 \left(  1-\sqrt{1-C_{\mbox{\tiny HST}}(U_1, U_2)} \right) \,. 
\end{equation}
Given that $C_{\mbox {\tiny LHST}}$ is bounded by $C_{\mbox {\tiny HST}}$ via~\eqref{eq:LHSTvsHSTbound}, the fast-forwarded $C_{\mbox {\tiny LHST}}$ can similarly be bounded as 
\begin{equation}\label{eq:fastforwardCLHST}
    \small{1-\sqrt{1-C_{\mbox{\tiny LHST}}(U^N, V^N)} \leq N^2 \left(  1-\sqrt{1-n \, C_{\mbox{\tiny LHST}}(U, V)} \right)} \ , 
\end{equation}
where we assume $n \, C_{\mbox{\tiny LHST}} \leq 1$ and we chose $U_1 = U$ and $U_2 = V$. Equation~\eqref{eq:fastforwardCLHST} is the exact version of \eqref{eqn:subquadratic} in the main text. Specializing to the case where the cost function $C_{\mbox{\tiny LHST}}$ is small, \eqref{eq:fastforwardCLHST} becomes \eqref{eqn:subquadratic}, i.e.,
\begin{equation}
    	C_{\mbox {\tiny LHST}}(U^N, V^N) \lessapprox  n \, N^2 \, C_{\mbox {\tiny LHST}}(U, V) \,.
\end{equation}

\subsection{An operational termination condition}\label{ap:TermCondition}

The VFF cost function is operationally meaningful by virtue of its relation to the average-case diagonalization error. Specifically, it can be shown~\cite{NielsonPLA2002, HorodeckiPRA1999} that
\begin{equation}\label{eq:CHSTtoAverageFid}
    C_{\mbox {\tiny HST}}(U_1, U_2) = \frac{d+1}{d}(1 -  \overline{F}(U_1,U_2))
\end{equation}
where
\begin{equation}
     \overline{F}(U_1, U_2) := \int_\psi | \ip{\psi(U_1)}{\psi(U_2)} |^2 d\psi \; 
\end{equation}
is the average fidelity over the Haar distribution. Therefore, from Eq.~\eqref{eq:LHSTvsHSTbound}, $C_{\mbox {\tiny LHST}}$ upper bounds the average fidelity as follows 
\begin{equation}\label{eq:CLHSTtoAverageFid}
    C_{\mbox {\tiny LHST}}(U_1, U_2) \geq \frac{d+1}{n d}(1 -  \overline{F}(U_1,U_2)) \; .
\end{equation}
This relation enables us to bound the average simulation error and hence provide a termination condition for VFF.

To derive a termination condition, we start from the bound on the total simulation error, Eq.~\eqref{eqnTriangleInequality2}, written in terms of the Hilbert-Schmidt norm ($p=2$) and take the minimum of the total simulation error and diagonalization errors over global phases applies to $V$ to remove their arbitrary phase dependence:
\begin{align}
  \tilde{\epsilon}_2\left(e^{-iHT}, V(\alv, T ) \right) \leq   N  \left( \epsilon_{2}^{\mbox{\tiny TS}}(\Delta t)  +
 \tilde{\epsilon}_2\left(U(\Delta t ), V(\alv, \Delta t )\right) \right) \,.
\end{align}
We can then rewrite this expression in terms of the average fidelity of the simulation,
\begin{equation}
        \overline{F}(T) := \overline{F}(e^{-iH T}, V(\alv, T )) \,,
\end{equation}
and the cost function $C_{\mbox{\tiny HST}}$ using~\eqref{eq:POTQtoHST} and~\eqref{eq:CHSTtoAverageFid}:
\begin{equation}
\begin{aligned}
G(T) \leq  \frac{N}{\sqrt{2d}} \epsilon_{2}^{\mbox{\tiny TS}}(\Delta t) + N \sqrt{1-\sqrt{1-C_{\mbox{\tiny HST}}^{\mbox{\tiny VFF}}(\Delta t)  }} \,. 
\end{aligned}
\end{equation}
Here, to simply the expression, we have defined
\begin{equation}
    G(T) := \sqrt{1-\sqrt{1-\frac{d+1}{d}(1-\overline{F}(T))}}\,,
\end{equation}
and $C_{\mbox{\tiny HST}}^{\mbox{\tiny VFF}}(\Delta t):= C_{\mbox{\tiny HST}}(U(\Delta t ), V(\alv, \Delta t ))$. As the operator norm is typically used for Trotter error analysis in the quantum simulation literature~\cite{Wecker2014, Poulin2015}, we rewrite $\epsilon_{2}^{\mbox{\tiny TS}}(\Delta t)$ in terms of the operator norm using equivalence relation $|| X ||_2 \leq \sqrt{d}|| X ||_\infty$, 
\begin{equation}\label{eq:AvergeFidCostFunctTrotBound1}
 G(T) \leq  \frac{N}{\sqrt{2}} \epsilon_{\infty}^{\mbox{\tiny TS}}(\Delta t) + N \sqrt{1-\sqrt{1-C_{\mbox{\tiny HST}}^{\mbox{\tiny VFF}}(\Delta t)  }} \; .
\end{equation}
Finally, $C_{\mbox{\tiny HST}}$ is upper bounded by $n C_{\mbox{\tiny LHST}}$ and therefore
\begin{equation}\label{eq:AvergeFidCostFunctTrotBound2}
 G(T) \leq  \frac{N}{\sqrt{2}} \epsilon_{\infty}^{\mbox{\tiny TS}}(\Delta t) + N \sqrt{1-\sqrt{1-n C_{\mbox{\tiny LHST}}^{\mbox{\tiny VFF}}(\Delta t)  }} \,, 
\end{equation}
assuming that $n C_{\mbox{\tiny LHST}}^{\mbox{\tiny VFF}}(\Delta t)\leq 1$. Re-arranging terms and denoting by $\epsilon(\Delta t) := \frac{1}{\sqrt{2}}\epsilon_{\infty}^{\mbox{\tiny TS}}(\Delta t) + \sqrt{1-\sqrt{1-nC_{\mbox{\tiny LHST}}^{\mbox{\tiny VFF}}(\Delta t) }}$ we get that whenever $\epsilon(\Delta t) \leq 1/N$ then the average fidelity is bounded by
\begin{align}
 \overline{F}(T) &\geq \frac{1}{d+1} + \frac{d}{d+1} \left(1-N^{2}\epsilon(\Delta t)^2\right)^{2} \\
 & = 1-\frac{N^2\epsilon(\Delta t)^2 d}{d+1}\left(2-N^2\epsilon(\Delta t)^2\right)\\
 &\geq 1 - \frac{N^2 d}{d+1} \, 2 \epsilon(\Delta t)^2 \,.\label{eqn:FinalBoundAvgF} 
\end{align}

One can get a more compact (but weaker) lower bound on the average fidelity by observing that \begin{align}
    \sqrt{1-\sqrt{1-nC_{\mbox{\tiny LHST}}^{\mbox{\tiny VFF}}(\Delta t)}}&\leq \sqrt{n C_{\mbox{\tiny LHST}}^{\mbox{\tiny VFF}}(\Delta t)} \;
\end{align} 
which holds whenever $n C_{\mbox{\tiny LHST}}^{\mbox{\tiny VFF}}(\Delta t)\leq 1$. Therefore we get an upper bound on $\epsilon(\Delta t)\leq \frac{1}{\sqrt{2}}\epsilon^{\mbox{\tiny TS}}_{\infty}(\Delta t) + \sqrt{nC_{\mbox{\tiny LHST}}^{\mbox{\tiny VFF}}(\Delta t)}$ so that \eqref{eqn:FinalBoundAvgF} becomes:
\begin{align}
    \overline{F}(T)&\geq 1- \frac{N^2\,d}{d+1}\left(\epsilon^{\mbox{\tiny TS}}_{\infty}(\Delta t) + \sqrt{2n C_{\mbox{\tiny LHST}}^{\mbox{\tiny VFF}}(\Delta t)}\right)^2\,.
\end{align}

After a successful optimization procedure, $n C_{\mbox{\tiny LHST}}^{\mbox{\tiny VFF}}(\Delta t)$ is expected to be small and Eq.~\eqref{eqn:FinalBoundAvgF} reduces to
\begin{equation}\label{eq:AvergeFidCostFunctTrotBoundApprox}
\begin{aligned}
\overline{F}(T) \gtrapprox 1 - \frac{N^2 d}{d+1}  \left(\epsilon_{\infty}^{\mbox{\tiny TS}}(\Delta t)  + \sqrt{ n C_{\mbox{\tiny LHST}}^{\mbox{\tiny VFF}}(\Delta t)} \right)^2   \; .
\end{aligned}
\end{equation}

Given a fixed initial Trotter error and for a target fast-forwarding time and simulation fidelity, Eq.~\eqref{eq:AvergeFidCostFunctTrotBound2}, and Eq.~\eqref{eq:AvergeFidCostFunctTrotBoundApprox} its simplified approximate variant, prescribe the cost function which must be surpassed before terminating the optimization loop.
Specifically, the termination condition is $C_{\mbox{\tiny LHST}}^{\mbox{\tiny VFF}} \leq C_{\mbox{\tiny Threshold}}$, with
\begin{equation}\label{eq:CostThreshold}
\begin{aligned}
C_{\mbox{\tiny Threshold}} &= \frac{1}{n} \left( 1 - \left( 1 - \left( \frac{1}{N} \sqrt{1-\sqrt{1-\frac{d+1}{d}(1-\overline{F}(T))}} - \frac{1}{\sqrt{2}} \epsilon_{\infty}^{\mbox{\tiny TS}}(\Delta t) \right)^2 \right)^2\right) \\ 
& \approx \frac{1}{n} \left(\frac{1}{N}\sqrt{\frac{d+1}{ d}(1-\overline{F}(T))} -\epsilon_{\infty}^{\mbox{\tiny TS}}(\Delta t) \right)^2 \; .  
\end{aligned}
\end{equation}

\section{Estimation of energy eigenvalues}\label{ap:Eigen}

Suppose $\{\lambda_1^{\rm{exact}},...,\lambda_d^{\rm{exact}}\}$ are the energies corresponding to the target unitary $U$, and $\{\lambda_1,...\lambda_d\}$ estimates extracted (using for example a time-series analysis) from the approximate unitary $V$ obtained from the VFF algorithm. Then, using the Hoffmann-Wielandt theorem, there is an ordering of the approximate energies (i.e a permutation $\sigma$) so that
\begin{equation}
    \sum_{i} |e^{i\lambda_{i}^{\rm{exact}}}- e^{i\lambda_{\sigma(i)}}|^2 \leq ||U-V||_{2}^2.
\end{equation}
Expanding the above and introducing the additional (fixed) arbitrary global phase $\phi_0$ we get
\begin{align}
    &\sum_{i} 2(1- \cos{(\lambda_i^{\rm{exact}} - \lambda_{\sigma(i)} + \phi_0)})\leq ||U-e^{i\phi_0}V||_{2}^{2}, 
\end{align}
where $\phi_0$ achieves $\min_{\phi} ||U-e^{i\phi}V||_{2} = \sqrt{2d(1-\sqrt{1-C_{HST}(U,V)})}$, as explained in \ref{ap:CertBounds}. Further, the upper bound can be related to the local cost function $C_{\mbox{\tiny LHST}}^{\mbox{\tiny VFF}}$.
\begin{align}
    &\sum_{i} 2(1- \cos{(\lambda_i^{\rm{exact}} - \lambda_{\sigma(i)} + \phi_0)})\leq \sqrt{2d(1-\sqrt{1-nC_{\mbox{\tiny LHST}}^{\mbox{\tiny VFF}}(U,V)})} , 
\end{align}
For small values of the cost function $C_{\mbox{\tiny LHST}}^{\mbox{\tiny VFF}}$ the upper bound will be small resulting in $\lambda_i^{\rm{exact}}\approx \lambda_{\sigma(i)}$.

\section{Cost function gradient derivation}\label{ap:GradientCalc}

Here we provide the derivation of the partial derivative of $C_{\mbox{\tiny LHST}}^{\mbox{\tiny VFF}}$ with respect to $\theta_k$ in \eqref{eq:CostEigenvectorGradient}. To emphasize the dependence on $\thv$, we write 
\begin{equation}
\begin{aligned}
    C_{\mbox{\tiny LHST}}^{\mbox{\tiny VFF}} = 1 - F(\vec{\theta})\,,\qquad \text{with} \qquad 
    F(\vec{\theta}) = \Tr[X (U^* \otimes W D W^\dagger) \dya{\Phi^+} (U^T \otimes W D^\dagger W^\dagger)], 
\end{aligned}
\end{equation}
where we have 
$X = \frac{1}{n} \sum_{j =1}^n \dya{\Phi^+}_{A_j, B_j} \otimes \id_{\overline{A}_j, \overline{B}_j}$ with $\ket{\Phi^+} = \frac{1}{\sqrt{2}} \left(\ket{00} + \ket{11}\right)$.
Taking the partial derivative of the cost function with respect to an angle $\theta_k$ gives
\begin{equation}\label{eq:DifF}
\begin{aligned} 
     \frac{\partial C_{\mbox{\tiny LHST}}^{\mbox{\tiny VFF}}}{\partial \theta_k} = &- \Tr \left[X \left(U^* \otimes \frac{\partial W}{\partial \theta_k} D W^\dagger \right) \dya{\Phi^+} \left(U^T \otimes W D^\dagger W^\dagger \right) \right] \\ 
     &- \Tr \left[X \left(U^* \otimes W D W^\dagger \right) \dya{\Phi^+} \left(U^T \otimes W D^\dagger \frac{\partial W^\dagger}{\partial \theta_k} \right) \right] \\ 
     &- \Tr \left[X \left(U^* \otimes W D \frac{\partial W^\dagger}{\partial \theta_k}\right) \dya{\Phi^+} \left(U^T \otimes W D^\dagger W^\dagger \right) \right] \\
     &- \Tr \left[X \left(U^* \otimes W D W^\dagger \right) \dya{\Phi^+} \left(U^T \otimes \frac{\partial W}{\partial \theta_k} D^\dagger W^\dagger \right) \right] . 
\end{aligned}
\end{equation}
The eigenvector operator, $W$, consists of products of Pauli rotations and can be decomposed as 
\begin{equation}
W = W_L  \exp \left(- \frac{ i \theta_k \sigma_k }{2} \right)  W_{R'} \equiv W_L W_R
\end{equation}
where the operators $W_L$ and $W_{R'}$ consist of all Pauli rotations to the left and right of the $\sigma_k$ rotation respectively and we have defined $W_R = \exp(- i \theta_k \sigma_k/2) W_{R'} $ for convenience. It follows that the differential of $W$ with respect to $\theta_k$ takes the form 
\begin{equation}
\frac{\partial W}{\partial \theta_k} = - \frac{1}{2} i W_L \sigma_k W_{R} \; , 
\end{equation}
which on substituting into Eq.~\eqref{eq:DifF} gives  \begin{equation}\label{eq:SubDifF}
\begin{aligned}
     \frac{\partial C_{\mbox{\tiny LHST}}^{\mbox{\tiny VFF}}}{\partial \theta_k} 
     = \frac{i}{2} \bigg( \Tr \left[X \left(U^* \otimes W_L \right) [ \sigma_k  , \rho_1]\left(U^T \otimes W_L^\dagger  \right) \right]  
- \Tr \left[X \left(U^* \otimes W D  W_R^\dagger\right) [\sigma_k,\rho_2]  \left(U^T \otimes W_R D^\dagger W^\dagger \right) \right] \bigg) \; ,
\end{aligned}
\end{equation}
where we have defined
\begin{align}
    &\rho_1 = W_{R} D W^\dagger \dya{\Phi^+} W D^\dagger W_{R}^\dagger \ \ \ \text{and} \\  &\rho_2 = W_L^\dagger \dya{\Phi^+} W_L \ . 
\end{align}
Eq.~\eqref{eq:CostEigenvectorGradient} is now obtained directly from Eq.~\eqref{eq:SubDifF} via the following identity, which holds for any state $\rho$,
\begin{align}
    i  [ \sigma_k  , \rho] &= e^{ i  \sigma_k \pi /4 } \rho e^{- i  \sigma_k \pi /4 } - e^{ -i  \sigma_k \pi /4 } \rho e^{ i  \sigma_k \pi /4 } \,.
\end{align}

{\color{black}
\section{Fast forwarding of Hamiltonians: From Asymptotic Time Complexity to the NISQ Era}\label{ap:FastForwarding}
The aim of this section is to give a brief summary of known complexity results and open questions regarding Hamiltonians that allow fast forwarding and to discuss their implications for near-term algorithms like VFF. 

A Hamiltonian $H$ is said to be fast forwardable if it can be simulated for time $T$ with computational resources much smaller than $T$. The work of Atia and Aharonov \cite{atia2017fast} establishes an equivalence between fast forwarding of Hamiltonians and performing efficient energy measurements with high precision. More specifically~\cite{atia2017fast}, and as sketched in Fig.~\ref{fig:FastForwarding}, a family of (normalised) Hamiltonians $H(n)$ that acts on an increasing number of qubits $n$ is fast-forwardable if for all $T \leq t(n)$ there is a circuit $U$ with polynomial size $d(n)$ such that $||e^{-iHT} - U||_{\infty}\leq \mathcal{O}(1/\mathrm{poly}(n))$ and $t(n)$ grows asymptotically faster than $d(n)$. In particular, whenever time grows exponentially in the number of qubits $t(n) = \mathcal{O}(2^{\Omega(n)})$, Hamiltonians satisfying this definition are said to be exponentially fast-forwardable. Specific classes of Hamiltonians have been shown to be exponentially fast-forwardable \cite{atia2017fast} including commuting Hamiltonians, quadratic fermionic systems and a familly of Hamiltonians related to Shor's algorithm. Recent geometric results \cite{balasubramanian2020quantum} suggest that the Sachdev-Ye-Kitaev model of chaotic systems can also be fast-forwarded. 

More generally, a connection between diagonalization of Hamiltonians and exponential fast forwarding has been established. Following \cite{atia2017fast}, a Hamiltonian $H$ is said to be \emph{quantum diagonalizable} if there exists a polynomial sized diagonalizing unitary $V$, and polynomial sized diagonal unitary $D$ such that $VHV^{\dagger} = D$. Any such Hamiltonian admits exponential fast forwarding. However, it remains an open question to investigate if the set of fast-forwardable Hamiltonians is larger than the set of quantum diagonalizable Hamiltonians \cite{novo2019quantum}. These types of questions are similar to proving tighter lower bounds for the circuit complexity of large time simulations restricted to specific classes of Hamiltonians. Further, we emphasize that fast-forwardable Hamiltonians can generate highly non-classical behaviour. For example, in \cite{novo2019quantum} a diagonalizable Hamiltonian for which $V$ is constructed out of IQP circuits gives a quantum advantage proposal for the task of energy sampling.

The no-fast-forwarding theorem has been established in several Hamiltonian models of computation. For the unknown Hamiltonian model  \cite{atia2017fast}, the no-fast-forwarding theorem states that there is no generic method to exponentially fast forward an arbitrary 2-sparse row computable Hamiltonian. A similar result for the query model (where access to the Hamiltonian is achieved by querying its matrix elements) implies that for every time $T$ there exists a (normalised, row-computable sparse) Hamiltonian whose simulation with finite precision for time $T$ requires at least $\Omega(T)$ queries.  These results establish the existence of Hamiltonians for which no generic fast forwarding occurs, without a constructive way to build non-fast forwardable Hamiltonians. As a consequence, large classes of Hamiltonians that allow universality will necessarily include families for which there is no generic way of fast-forwarding. As an example, nearest-neighbour Heisenberg interactions on a lattice give rise to a universal class of Hamiltonians.

\medskip

\paragraph*{What are the implications of these results to VFF?}
It is highly important to emphasize that the notion of fast forwarding technically refers to a \emph{family} of Hamiltonians acting on an increasing number of qubits $n$, so it can be misleading to refer to fast forwarding of a Hamiltonian with a fixed number of qubits. Thus the question at hand is whether/how the no-fast-forwarding theorems restrict the use of VFF to simulate families of Hamiltonians. 

To answer this we first note that asymptotic complexity results cannot be directly applied to finite size $n$ and finite depth $d$ experiments. Therefore, as it is the small and intermediary scale behaviour that determines the feasibility and practical applicability of VFF experiments for the foreseeable future, it follows that the questions of whether a family of Hamiltonians is asymptotically fast-forwardable and whether VFF can be used to perform simulation beyond the coherence time are largely independent. As sketched in Fig.~\ref{fig:FastForwarding} it is possible to conceive of families of Hamiltonians that are asymptotically non-fast-forwardable and yet, for a regime of finite $n$, allow for simulation beyond the coherence time. Conversely, there may exists families of Hamiltonians which are fast-forwardable but do not permit simulation beyond the coherence time, due to an unfavourable scaling at low $n$. In this sense, the complexity-theoretic no-fast-forwarding results are not expected to place restrictions on the practical use of VFF for NISQ simulations.

\begin{figure}[t]
\centering
{\includegraphics[width=0.45\linewidth]{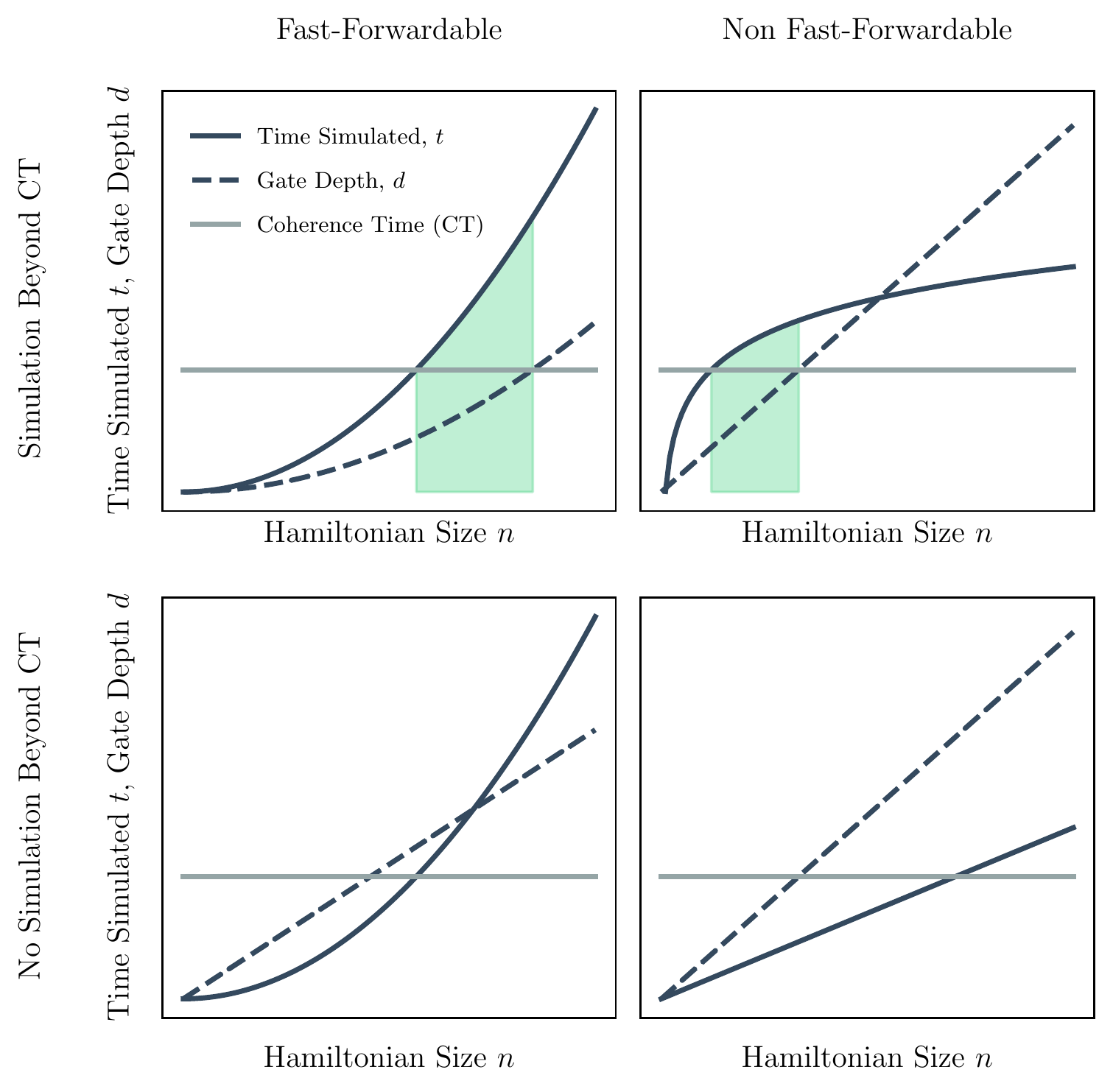}}
\caption{Schematic highlighting the independence of the ability to fast-forward a family of Hamiltonians and the ability to use VFF to simulate a subset of those Hamiltonians beyond the coherence time. The blue lines indicate the time $t$ that can be simulated (solid) and the corresponding gate depth $d$ required to do so (dashed), as a function of the number of simulated qubits $n$. The two figures in the left column indicate fast-forwardable Hamiltonians with $t(n)$ growing asymptotically faster than $d(n)$ (i.e. the solid blue line is above the dashed blue line for large $n$). Correspondingly, the two figures in the right column sketch $t(n)$ and $d(n)$ for non fast-forwardable Hamiltonians. The grey line indicates the coherence time of the quantum computer, that is the gate depth after which quantum simulations are no longer possible due to the build up of errors and decoherence. The green shading indicates the region in which simulation beyond the coherence time is possible. This is the scale on which the time simulated $t$ can be greater than the coherence time of quantum computer but the required gate depth $d$ to perform that simulation does not exceed the coherence time (i.e. the region in which the solid blue line is above, and the dashed blue line is below, the grey line).}
\label{fig:FastForwarding} 
\end{figure} 

From a theoretical standpoint, the bottleneck for VFF comes from the requirement that the unitaries $W$ and $D$ have a-priori a fixed depth. This will affect the implementation of VFF by limiting the theoretical minimum for the optimisation error $\epsilon_{ML}$.  From an implementation point of view, the limitations of VFF were addressed in the main text with the error analysis and termination condition.  The finite depth cut off will introduce an error $\delta = \mathrm{min}_{W,D} ||e^{-iH\Delta t} - WDW^{\dagger}|| $, where the minimisation occurs over all those unitaries $W$ and $D$ for which the total depth of $WDW^{\dagger}$ is at most the coherence time of the hardware device considered (i.e some fixed value). Naturally, $\delta$ will be a lower bound for the machine learning error from the optimisation step of the VFF algorithm $\delta\leq \epsilon_{ML}$. The fast forwardability of a family of Hamiltonians $H$ acting on an increasing number of qubits $n$ will not necessarily be enough to ensure that $\delta$ is small and scales as $\mathcal{O}(1/\mathrm{poly}(n))$ (i.e with a similar approximation error as in the definition of fast forwardability). The same holds for the possibly weaker condition of quantum diagonalisability. 

However, we could expect that the error in approximating the diagonalization with a fixed depth ansatz will exhibit better large $n$  scaling for quantum diagonalizable Hamiltonians than for those whose diagonalisation requires exponential resources. Thus the complexity-theoretic results may place fundamental limitations on the scalability of VFF for certain Hamiltonians as we reach the fault tolerant regime. 

That is not to say that Hamiltonians which are not quantum diagonalizable cannot have, for small $n$, a good fixed-depth ansatz approximation of their diagonalisation. We expect this to be true in particular of those Hamiltonians for which the exponential scaling of resources becomes dominant for much larger $n$. By contrast it will be possible that the fixed-depth ansatz approximation could not be suitable for certain Hamiltonians whose diagonalisation has a high overhead in constant factors, even if it has a polynomial asymptotic scaling.

Another issue that will affect the minimal theoretically achievable error $\delta$ for VFF is the expressivity of the fixed-depth ansatz. For a given depth, can the unitaries $W$ and $D$ that give the minimal $\delta$  be exactly represented by some parametrisation for the hardware-efficient ansatz (or more sophisticated ones such as UCC)? Generally, this will not necessarily be the case and this will be reflected in a higher value for $\delta$ with the restriction that one must minimise over particular fixed depth structures.

}

\end{document}